\documentclass[11pt,a4paper]{amsart}
\usepackage{a4wide}
\usepackage{amsmath,amssymb,amsthm,mathtools,mathrsfs,booktabs}
\usepackage{url}
\usepackage[colorlinks=true,linkcolor=blue,citecolor=blue,urlcolor=blue]{hyperref}

\title[Coherent risk from finite models]{Coherent risk from finite models}
\author[T.~Kania]{Tomasz Kania}
\address[T.~Kania]{Mathematical Institute\\Czech Academy of Sciences\\\v Zitn\'a 25 \\115 67 Praha 1\\Czech Republic  and  Institute of Mathematics and Computer Science\\ Jagiellonian University\\ {\L}ojasiewicza 6, 30-348 Krak\'{o}w, Poland
}
\email{kania@math.cas.cz, tomasz.marcin.kania@gmail.com}
\thanks{RVO: 67985840.}
\date{}
\hypersetup{pdftitle={Coherent risk from finite models},pdfauthor={Tomasz Kania}}

\theoremstyle{plain}
\newtheorem{theorem}{Theorem}[section]
\newtheorem{proposition}[theorem]{Proposition}
\newtheorem{lemma}[theorem]{Lemma}
\newtheorem{corollary}[theorem]{Corollary}
\newtheorem{maintheorem}{Theorem}

\theoremstyle{definition}
\newtheorem{definition}[theorem]{Definition}
\newtheorem{example}[theorem]{Example}
\theoremstyle{remark}
\newtheorem{remark}[theorem]{Remark}
\newcommand{\R}{\mathbb R}
\newcommand{\N}{\mathbb N}
\newcommand{\E}{\mathbb E}
\newcommand{\Pbb}{\mathsf P}
\newcommand{\ind}{\mathbf 1}
\newcommand{\Linf}{L^\infty}
\newcommand{\ES}{\operatorname{ES}}
\newcommand{\st}{\operatorname{st}}
\newcommand{\Wone}{W_1}
\newcommand{\Zrho}{\mathcal Z_\rho}
\newcommand{\Pa}{\mathcal P[\underline x,\overline x]}
\newcommand{\rep}{\mathsf R}
\newcommand{\agg}{\mathsf A}
\newcommand{\spl}{\mathsf J}

\begin{document}
\begin{abstract}
Coherent estimators at different sample sizes need not describe the same
population risk. We prove that dyadic estimators do so exactly when duplicating
every scenario leaves the estimate unchanged, and recover the population dual
set from the finite models. The risk of a Bernoulli loss determines the exact
Wasserstein, Kolmogorov and total-variation moduli on a common bounded support.
It also characterises uniform estimability, with an impossibility bound for
arbitrary randomised estimators. An intrinsic floor on tail levels gives an
explicit optimal Lipschitz constant and a uniform approximation defect;
every deterministic vanishing floor preserves pointwise strong consistency.
A nonstandard construction supplies compatible estimators at all finite sample
sizes, explicit block-sum weights, and a formula for the worst-case mass created
by infinitesimal tail levels. Finally, we derive uniform regret bounds for
data-driven risk-averse optimisation and specialise them to stop-loss retention.
Exact binomial calculations compare regularised and original estimators against
the same population target.
\end{abstract}
\maketitle

\section{Introduction}

A risk estimator assigns a capital requirement to an observed profit-and-loss
sample. Aichele, Cialenco, Jelito and Pitera \cite{ACJP} impose coherence on
this sample map itself: monotonicity, cash additivity, positive homogeneity
and subadditivity must hold for every pair of sample vectors. Their framework
transfers the financial axioms of \cite{ADEH} to the estimation procedure,
and characterises law-invariant coherent risk estimators by suprema of
decreasingly weighted order statistics. Under comonotonic additivity one
weight vector suffices. We adopt their terminology and recall the relevant
definitions in Section~\ref{subsec:cre}, writing CRE for coherent risk estimator.

Coherence at each sample size leaves compatibility across sample sizes open.
Requiring that duplication of every equally weighted observation preserve
the estimate gives a precise extension criterion and a finite test for
uniform estimability. Duplication represents a refinement of the scenario
table: repeating recorded observations does not create independent data.
The compatible family determines a Bernoulli risk profile, which gives exact
transport stability and an intrinsic regularisation of tail weights. These
bounds also control data-driven risk-averse decisions. A complementary
nonstandard construction begins with one internal coherent estimator and
identifies its population risk, finite weights and limiting tail mass.

Coherence, replication consistency and statistical accuracy have distinct
r\^oles. Neither coherence nor replication consistency implies unbiasedness
or guarantees observation of a rare loss. This distinction between a risk
functional and its estimation procedure is also central to
Cont, Deguest and Scandolo \cite{CDS}. We retain the fixed-sample-size
axioms of \cite{ACJP} throughout, and state the additional consistency
properties separately.

The notation is developed in Section~\ref{sec:preliminaries}. We work on
$[0,1]$ with completed Lebesgue measure $\lambda$.
For $m\geqslant0$, put $n_m=2^m$ and
$I_{j,m}=((j-1)/n_m,j/n_m]$. Null endpoints are immaterial. Define
\[
 S_mx=\sum_{j=1}^{n_m}x_j\ind_{I_{j,m}},\qquad
 P_mX=\left(n_m\int_{I_{j,m}}X\,d\lambda\right)_{j=1}^{n_m},
 \qquad \E_mX=S_mP_mX.
\]
Write $\Delta_n=\{a\in[0,1]^n:\sum_ja_j=1\}$, and let
$\rep_mx=(x_1,x_1,\ldots,x_{n_m},x_{n_m})$ be the duplication map.
A functional on $\R^n$ is called law invariant when it is invariant under
coordinate permutations; its scenarios have equal probabilities.

\begin{maintheorem}\label{thm:reconstruction}
Let $r_m:\R^{n_m}\to\R$, $m\geqslant0$, be law-invariant CREs.
There exists a law-invariant coherent risk measure
$\rho:\Linf[0,1]\to\R$ such that
\begin{equation}\label{eq:restrictions}
 \rho(S_mx)=r_m(x)\qquad(m\geqslant0, x\in\R^{n_m})
\end{equation}
if and only if
\begin{equation}\label{eq:replication}
 r_{m+1}(\rep_mx)=r_m(x)\qquad(m\geqslant0, x\in\R^{n_m}).
\end{equation}
This extension is unique among all law-invariant coherent risk measures,
has the Fatou property, and satisfies
\begin{equation}\label{eq:reconstruction}
 \rho(X)=\lim_{m\to\infty}r_m(P_mX)=\sup_{m\geqslant0}r_m(P_mX),
\end{equation}
where the sequence is non-decreasing. Let
\[
 D_m=\{a\in\Delta_{n_m}:-a\cdot x\leqslant r_m(x)
                         \text{ for every }x\in\R^{n_m}\},\qquad
 \mathcal C_m=\left\{n_m\sum_{j=1}^{n_m}a_j\ind_{I_{j,m}}:a\in D_m\right\}.
\]
Then $\mathcal C_m\subseteq\mathcal C_{m+1}$, and the canonical dual set is
\begin{align}
 \Zrho
 &=\overline{\bigcup_{m\geqslant0}\mathcal C_m}^{\,L^1}
 \label{eq:dual-closure}\\
 &=\left\{Z\in L^1_+:\E Z=1,
       \left(\int_{I_{j,m}}Z\,d\lambda\right)_{j=1}^{n_m}\in D_m
       \text{ for every }m\right\}.
 \label{eq:dual-moments}
\end{align}
Every finite dual set is recovered exactly:
\begin{equation}\label{eq:dual-projection}
 D_m=\left\{\left(\int_{I_{j,m}}Z\,d\lambda\right)_{j=1}^{n_m}:Z\in\Zrho\right\}.
\end{equation}
\end{maintheorem}

Extension results for risk measures provide the relevant context.
Filipovi\'c and Svindland \cite{FilipovicSvindland} obtain unique
law-invariant lower-semicontinuous convex extensions from $L^\infty$ to
$L^1$, allowing infinite values. Koch-Medina and Munari
\cite{KochMedinaMunari} characterise extensions to $L^p$ that preserve
finiteness and continuity and relate them to qualitative robustness.
Rahsepar and Xanthos \cite{RahseparXanthos} extend Fatou,
dilatation-monotone functionals from subsets of $L^1$ containing every
simple random variable. Here the initial domain consists only of the
functions constant on a finite dyadic partition. It does not contain all
simple random variables. Theorem~\ref{thm:reconstruction} supplies a local
criterion on this domain and recovers each finite polar exactly from the
population dual set.

The distributional extension theorem of Noyan and Rudolf
\cite[Section~5]{NoyanRudolf} assumes a family containing $\delta_0$ and
closed under non-negative scaling, constant shifts and sums under every
coupling of its marginal laws. Dyadic empirical laws fail the last
condition. For irrational $r\in(0,1/2)$ choose events $A,C$ with
$\lambda(A)=\lambda(C)=1/2$ and $\lambda(A\cap C)=r$; for instance,
$A=(0,1/2]$ and $C=(0,r]\cup(1/2,1-r]$. Each indicator has a dyadic
empirical law, whereas
\[
 \mathcal L(\ind_A+\ind_C)
 =r\delta_0+(1-2r)\delta_1+r\delta_2
\]
has irrational masses. Consequently their closure hypothesis does not
apply directly to this domain. The finite Kusuoka representation used
below remains classical \cite{Kusuoka,NoyanRudolf,Ugurlu}.

For the rest of the introduction, assume \eqref{eq:replication} and let
$\rho$ be the extension. Its associated Bernoulli profile is
\begin{equation}\label{eq:profile-definition}
 c_\rho(t)=\rho(-\ind_A),\qquad \lambda(A)=t,\quad 0\leqslant t\leqslant1.
\end{equation}
Law invariance makes the definition independent of $A$. This is the associated
distortion function of Kr\"atschmer, Schied and Z\"ahle
\cite[eq.~(5)]{KSZ2015}; we use it to obtain an exact modulus.
Write $\Pa$ for all probability laws supported in $[\underline x,\overline x]$, where $\underline x<\overline x$,
and set $B=\overline x-\underline x$. For a law $P$, the notation $\rho(P)$ means $\rho(X)$
for any $X$ with law $P$.

\begin{maintheorem}\label{thm:modulus}
For every $0\leqslant e\leqslant B$,
\begin{equation}\label{eq:exact-modulus}
 \sup_{\substack{P,Q\in\Pa\\\Wone(P,Q)\leqslant e}}
       |\rho(P)-\rho(Q)|=B c_\rho(e/B).
\end{equation}
Equality is attained by $Q=\delta_{\overline x}$ and
$P=(e/B)\delta_{\underline x}+(1-e/B)\delta_{\overline x}$.
If $d_m=r_m(-e_1)$, with $e_1$ the first coordinate vector in $\R^{n_m}$,
then $d_m=c_\rho(2^{-m})$. The following conditions are equivalent:
\begin{enumerate}
 \item $d_m\to0$;
 \item $c_\rho(t)\to0$ as $t\downarrow0$;
 \item $\Zrho$ is uniformly integrable;
 \item $\Zrho$ is weakly compact in $L^1$;
 \item $\rho$ has the Lebesgue property;
 \item $P\mapsto\rho(P)$ is uniformly $\Wone$-continuous on $\Pa$.
\end{enumerate}
\end{maintheorem}

Wasserstein perturbation bounds for risk measures are developed by Pichler
\cite{Pichler} and Prashanth and Bhat \cite{PrashanthBhat}. Continuity,
robustness and asymptotic estimation are studied in
\cite{CKS,CDS,GaoXanthos,KSZ2014,KSZ2017}. The classical equivalence
between the Lebesgue property and compactness of the dual set is used here
in its law-invariant, finite-calibration form. Formula~\eqref{eq:exact-modulus}
identifies the exact modulus on a \emph{common bounded support}; it applies
also to unbounded spectra and to discontinuous profiles.
The same profile gives exact Kolmogorov and total-variation moduli
(Corollary~\ref{cor:other-moduli}). Uniform
integrability in Theorem~\ref{thm:modulus} concerns the canonical set of
dual densities. This differs from the characterisation of uniformly
integrable sets of positions by auxiliary risk bounds in Huang and Wang
\cite{HuangWang}.

Let $\widehat P_n=n^{-1}\sum_{i=1}^n\delta_{X_i}$ for an independent sample
with common law $P$. A sequence of estimators is uniformly consistent on
$\Pa$ if, for every $\varepsilon>0$, its maximal probability of error
larger than $\varepsilon$ tends to zero over $P\in\Pa$.

\begin{maintheorem}\label{thm:statistics}
The conditions of Theorem~\ref{thm:modulus} are also equivalent to each of:
\begin{enumerate}
 \item the empirical plug-in estimators $\widehat\rho_n=\rho(\widehat P_n)$
       are uniformly consistent on $\Pa$;
 \item there exists a uniformly consistent sequence of estimators of $\rho(P)$
       on $\Pa$, allowing auxiliary randomisation.
\end{enumerate}
For $0<\delta<1$, put
\[
 u_{n,\delta}=\min\left\{1,\sqrt{\frac{\log(2/\delta)}{2n}}\right\}.
\]
For every $n\geqslant1$,
\begin{equation}\label{eq:concentration}
 \sup_{P\in\Pa}P^{\otimes n}
 \bigl(|\widehat\rho_n-\rho(P)|>Bc_\rho(u_{n,\delta})\bigr)\leqslant\delta.
\end{equation}
Conversely, for every Borel estimator
$T_n:[\underline x,\overline x]^n\times[0,1]\to\R$ in the randomised sense of
Section~\ref{sec:statistics}, every $0<t<1$, and every
$0<\varepsilon<Bc_\rho(t)/2$,
\begin{equation}\label{eq:testing-lower}
 \sup_{P\in\Pa}\Pbb_P\bigl(|T_n-\rho(P)|>\varepsilon\bigr)
 \geqslant\frac{(1-t)^n}{1+(1-t)^n}.
\end{equation}
Here $\Pbb_P$ includes the estimator's auxiliary randomisation, if present.
\end{maintheorem}

Pointwise strong consistency for each fixed bounded law is classical
\cite{WozabalWozabal}, including risks with $c_\rho(0+)>0$.
Theorem~\ref{thm:statistics} characterises when estimation can be uniform
over all laws on a common support. Bartl and Tangpi \cite{BartlTangpi}
give non-asymptotic plug-in bounds,
Bernoulli lower-bound examples, and examples of slow convergence despite the
Lebesgue property. The impossibility direction above applies to
\emph{all} estimators; the upper and lower bounds need not have matching rates.
Optimal expected-shortfall estimation under heavy tails and adversarial
contamination is studied by Bartl and Eckstein \cite{BartlEckstein}.
Our distribution-free statements concern bounded supports and do not
assert optimal rates or robustness to contaminated observations.
Functional delta methods and asymptotic distributions belong to the separate,
well-developed theory in \cite{BeutnerZahle,KSZ2015}; they are not needed
for the present results.

Our nonstandard construction reverses the direction of
Theorem~\ref{thm:reconstruction}: its input is one internal estimator.
Let $J$ be an infinite hyperinteger and set $H=J!$. Every positive
standard integer $n$ divides $H$. Write $\iota_{n,H}x$ for the internal
vector obtained by repeating each coordinate of $x\in\R^n$ exactly
$H/n$ times. Nonstandard terminology is recalled in
Section~\ref{subsec:nsa-prelim}.

\begin{maintheorem}\label{thm:shadow}
Let $\mathfrak r:{}^*\R^H\to{}^*\R$ be an internal law-invariant
coherent estimator, with the axioms interpreted internally. Then
\begin{equation}\label{eq:shadow-estimator}
 R_n(x)=\st\bigl(\mathfrak r(\iota_{n,H}x)\bigr),
 \qquad n\in\N,\quad x\in\R^n,
\end{equation}
defines law-invariant coherent risk estimators in the sense of
\cite{ACJP}. If $\mathsf D_kx$ repeats every coordinate $k$ times, then
\begin{equation}\label{eq:all-replication}
 R_{kn}(\mathsf D_kx)=R_n(x)\qquad(k,n\in\N).
\end{equation}
There is a unique law-invariant coherent risk $\rho$ on
$L^\infty[0,1]$ such that $R_n(x)=\rho(P_x)$ for every standard $n$.
If $\mathfrak r$ is internally comonotonically additive, so is every
$R_n$. In this case, if $(a_j)_{j=1}^H$ are its internal decreasing
order-statistic weights, the standard weights are explicitly
\begin{equation}\label{eq:shadow-weights}
 w_{i,n}=\st\left(\sum_{j=(i-1)H/n+1}^{iH/n}a_j\right),\qquad
 R_n(x)=-\sum_{i=1}^nw_{i,n}x_{i:n}.
\end{equation}
\end{maintheorem}

The standard content of Theorem~\ref{thm:shadow} is the all-sample-size
extension equivalence in Proposition~\ref{prop:all-sample-extension}:
replication-consistent CREs are precisely empirical restrictions of a population
risk. Every such family also admits an internal realisation. The nonstandard
construction gives the explicit block-sum weights in
\eqref{eq:shadow-weights}. For hyperfinite ES mixtures, Loeb pushforward
identifies ordinary mixing laws and, more particularly, the identity
\[
 c_\rho(0+)=\max_{\nu\in\Gamma}\nu(\{0\})
 =\lim_{\eta\downarrow0}\st\!\left(
       \max_{\mu\in\mathscr M}\sum_{k\leqslant\eta H}\mu_k\right)
\]
of Proposition~\ref{prop:loeb-tail-mass}. This expresses the obstruction to
uniform estimation directly in terms of internal tail mass. The harmonic
example \eqref{eq:harmonic-worst-case-cre} produces an atom at zero from
strictly positive internal levels. The resulting standard spectral and
mixed-superquantile formulae belong to the established theory
\cite{Acerbi,RockafellarRoyset}; the nonstandard contribution lies in
identifying their weights and limiting mass from one internal input.

The final main result provides a stability device intrinsic to the population
risk. For $\varepsilon\in(0,1]$, let $q_X^\varepsilon$ be the quantile obtained
by replacing $q_X$ on $(0,\varepsilon]$ by its average there, and set
$\rho^\varepsilon(X)=\rho(q_X^\varepsilon)$.
Equivalently, in any Kusuoka representation, replace each ES level $\alpha$
by $\max\{\alpha,\varepsilon\}$; Section~\ref{sec:regularisation} proves
that this does not depend on the chosen representation.

\begin{maintheorem}\label{thm:regularisation-main}
Let $\rho$ be a law-invariant coherent risk on $L^\infty[0,1]$, write
$c=c_\rho$ and $\kappa=c(0+)$, and let $0<\varepsilon\leqslant1$.
Then $\rho^\varepsilon$ is a law-invariant coherent risk with the Lebesgue
property and Bernoulli profile
\[
 c_\varepsilon(t)=
 \begin{cases}
  t c(\varepsilon)/\varepsilon,&0\leqslant t\leqslant\varepsilon,\\
  c(t),&\varepsilon\leqslant t\leqslant1.
 \end{cases}
\]
Its optimal $W_1$-Lipschitz constant on bounded laws is
$c(\varepsilon)/\varepsilon$. For every bounded $X$,
$\rho^\varepsilon(X)\uparrow\rho(X)$ as $\varepsilon\downarrow0$, and
\[
 B\kappa\leqslant
 \sup_{P\in\Pa}\{\rho(P)-\rho^\varepsilon(P)\}
 \leqslant Bc(\varepsilon).
\]
Thus the uniform approximation defect tends to $B\kappa$.
For every deterministic sequence $\varepsilon_n\in(0,1]$ with
$\varepsilon_n\to0$, the empirical estimates
$\rho^{\varepsilon_n}(\widehat P_n)$ converge almost surely to $\rho(P)$
at each bounded law $P$. They are uniformly consistent on $\Pa$ if and only
if $\kappa=0$.
\end{maintheorem}

Proposition~\ref{prop:tail-regularisation} and
Corollary~\ref{cor:reg-estimation} prove Theorem~\ref{thm:regularisation-main}.
Theorem~\ref{thm:decision-stability} then turns the exact modulus and the
one-sided bias into simultaneous confidence and regret bounds over a compact
decision set, including data-dependent choices of the floor. Its transport
formulation applies on any event controlling the empirical $W_1$ distance;
for a shared monotone scalar scenario, one distribution-free event controls
all decisions. Stop-loss retention supplies a concrete instance.

Section~\ref{sec:preliminaries} fixes conventions and proves the finite
representation tools. Sections~\ref{sec:reconstruction}--\ref{sec:statistics}
prove Theorems~\ref{thm:reconstruction}--\ref{thm:statistics}.
The nonstandard construction is developed in
Section~\ref{sec:nsa-construction}, and its regularisation in
Section~\ref{sec:regularisation}. Section~\ref{sec:retention} develops
data-driven optimisation and specialises its bounds to stop-loss contracts. Section~\ref{sec:catastrophe} gives exact
finite-sample calculations for a bounded catastrophe model.
Appendix~\ref{sec:hyperfinite} records the corresponding representation of
the conditional approximation in Theorem~\ref{thm:reconstruction}.

\section{Preliminaries}\label{sec:preliminaries}

\subsection{Risk conventions and duality}

All random variables are real valued and identified modulo almost-sure
equality. Positive values denote profits and negative values denote losses.
Consequently a risk measure decreases when the position improves.

\begin{definition}
A map $\rho:\Linf\to\R$ is coherent if, for $X,Y\in\Linf$,
$c\in\R$ and $s\geqslant0$,
\[
 \begin{aligned}
 X\leqslant Y&\ \Longrightarrow\ \rho(X)\geqslant\rho(Y),&
 \rho(X+c)&=\rho(X)-c,\\
 \rho(sX)&=s\rho(X),&
 \rho(X+Y)&\leqslant\rho(X)+\rho(Y).
 \end{aligned}
\]
It is law invariant if identically distributed positions have the same risk.
It has the Fatou property if
\[
 X_k\to X\text{ a.s.},\quad\sup_k\|X_k\|_\infty<\infty
 \quad\Longrightarrow\quad
 \rho(X)\leqslant\liminf_k\rho(X_k).
\]
It has the Lebesgue property if the last conclusion is instead
$\rho(X_k)\to\rho(X)$.
\end{definition}

Coherence implies convexity, $\rho(0)=0$,
$\rho(c)=-c$, and the elementary bound
\begin{equation}\label{eq:supnorm-lipschitz}
 |\rho(X)-\rho(Y)|\leqslant\|X-Y\|_\infty.
\end{equation}
Indeed $Y-\|X-Y\|_\infty\leqslant X\leqslant
Y+\|X-Y\|_\infty$, so monotonicity and cash additivity give both
inequalities. In particular, all CREs are continuous. On our standard atomless probability
space, every finite law-invariant convex risk measure has the Fatou property
by Jouini, Schachermayer and Touzi \cite[Theorem~2.2]{JST}.
Thus the Fatou property is automatic for the population coherent risks in this
paper; we retain it explicitly where it is used in a proof.

For a coherent risk with the Fatou property the standard dual representation
\cite{Delbaen,FS} is
\begin{equation}\label{eq:canonical-duality}
 \rho(X)=\sup_{Z\in\Zrho}\E[-XZ],\qquad
 \Zrho=\{Z\in L^1_+:\E Z=1,\ \E[-XZ]\leqslant\rho(X)
                                      \text{ for every }X\in\Linf\}.
\end{equation}
The set $\Zrho$ is non-empty, convex and norm closed. A supremum is used;
attainment is not assumed. Uniform integrability means
\[
 \lim_{M\to\infty}\sup_{Z\in\Zrho}\E[Z\ind_{\{Z>M\}}]=0.
\]
For densities of mass one this is equivalent to uniform absolute
continuity of the integrals on sets of small probability.

\subsection{Coherent risk estimators}\label{subsec:cre}

We use the sample-map terminology of \cite[Sections~2 and~3]{ACJP}.
It is essential that the axioms below apply to arbitrary deterministic
samples, before a sampling distribution or a target risk is specified.

\begin{definition}
A risk estimator at sample size $n$ is a Borel measurable map
$R_n:\R^n\to\R$. It is a \emph{coherent risk estimator} (CRE) if,
for $x,y\in\R^n$, $b\in\R$ and $s\geqslant0$,
\begin{align*}
 x\leqslant y&\ \Longrightarrow\ R_n(x)\geqslant R_n(y),
                         &&\text{monotonicity},\\
 R_n(x+b\ind)&=R_n(x)-b,  &&\text{cash additivity},\\
 R_n(sx)&=sR_n(x),         &&\text{positive homogeneity},\\
 R_n(x+y)&\leqslant R_n(x)+R_n(y), &&\text{subadditivity}.
\end{align*}
It is \emph{law invariant} if $R_n(x_\pi)=R_n(x)$ for every coordinate
permutation $\pi$.
\end{definition}

Here $\ind=(1,\ldots,1)$ and $x\leqslant y$ means coordinatewise order.
The argument for \eqref{eq:supnorm-lipschitz} applies verbatim to $R_n$,
so the four coherence axioms already imply continuity and measurability.
Law invariance refers to the equal-weight empirical law at a fixed $n$;
it imposes neither independence of observations nor compatibility between
different sample sizes.
We write $\Delta_n^\downarrow=\{w\in\Delta_n:w_1\geqslant\cdots\geqslant w_n\}$.

\begin{definition}
Vectors $x,y\in\R^n$ are \emph{comonotone} if
$(x_i-x_j)(y_i-y_j)\geqslant0$ for every $i,j$. A CRE is
\emph{comonotonically additive} if $R_n(x+y)=R_n(x)+R_n(y)$
for every such pair. An \emph{$L$-estimator} of risk is a map
\[
 L_w(x)=-\sum_{i=1}^nw_i x_{i:n},
\]
where the coefficients $w_i$ are deterministic. In this paper the
minus sign reflects the profit convention; coherence requires
$w\in\Delta_n^\downarrow$, as proved below.
\end{definition}

For a law-invariant population risk $\rho$, the \emph{empirical plug-in}
estimator is $R_n(x)=\rho(P_x)$, where
$P_x=n^{-1}\sum_{i=1}^n\delta_{x_i}$. These maps are CREs whenever
$\rho$ is coherent \cite[Section~2]{ACJP}: assigning $x$ to its
equal-cell step function is linear, order preserving and constant
preserving, so each axiom passes to the composition. A general CRE
need not be presented as the plug-in estimator of a prescribed risk.

\begin{definition}
A sequence $(R_n)$ is \emph{consistent} for $\rho$ on a class of laws
if $R_n(X_1,\ldots,X_n)\to\rho(P)$ in probability for each law $P$
in that class and an independent sample with common law $P$. It is
\emph{strongly consistent} if convergence holds almost surely.
It is \emph{replication consistent} if \eqref{eq:all-replication} holds.
\end{definition}

The statistical notions follow \cite[Section~4]{ACJP}; replication
consistency is an additional exact identity between sample maps.
Theorem~\ref{thm:reconstruction}
specifies when dyadic CREs arise from one population risk, complementing
the fixed-size representations and consistency results of \cite{ACJP}.

\subsection{Quantiles and spectra}

Write $F_P$ for the distribution function of a probability law $P$ and put
\[
 q_P(u)=\inf\{x\in\R:F_P(x)\geqslant u\},\qquad 0<u<1.
\]
The notation $q_X$ means $q_{\mathcal L(X)}$. For laws with finite first
moments, the one-dimensional transport identities are
\begin{equation}\label{eq:wasserstein}
 \Wone(P,Q)=\int_0^1|q_P(u)-q_Q(u)|\,du
           =\int_\R|F_P(x)-F_Q(x)|\,dx;
\end{equation}
see \cite{BobkovLedoux,Villani}. A spectrum is a non-negative,
non-increasing integrable function $\phi$ on $(0,1)$ with integral one.
With the profit convention, its risk functional is
\[
 \rho_\phi(X)=-\int_0^1q_X(u)\phi(u)\,du.
\]
Expected shortfall at \emph{lower-tail probability} $\alpha\in(0,1]$ is
\begin{equation}\label{eq:es}
 \ES_\alpha(X)=-\frac1\alpha\int_0^\alpha q_X(u)\,du
   =\sup_{\substack{0\leqslant Z\leqslant1/\alpha\\\E Z=1}}\E[-XZ].
\end{equation}
The dual equality follows by assigning mass $1/\alpha$ to the lowest
$\alpha$ fraction of $X$, splitting a level set if necessary. It proves
coherence. Also,
\begin{equation}\label{eq:es-w1}
 |\ES_\alpha(P)-\ES_\alpha(Q)|\leqslant\alpha^{-1}\Wone(P,Q).
\end{equation}
Thus each finite mixture of these expected shortfalls has the Lebesgue
property. General spectral risks are coherent as in \cite{Acerbi}.

For $x\in\R^n$, write $x_{1:n}\leqslant\cdots\leqslant x_{n:n}$ for
the ordered coordinates.

For the empirical law $P_x$, we use the continuous extension
$\ES_0(P_x)=-x_{1:n}$ of empirical expected shortfall.
For $0<\alpha<1$, with $j=\lfloor n\alpha\rfloor$, its exact formula is
\begin{equation}\label{eq:fractional-empirical-es}
 \ES_\alpha(P_x)
 =-\frac{\sum_{i=1}^j x_{i:n}+(n\alpha-j)x_{j+1:n}}
             {n\alpha},
 \qquad
 \ES_1(P_x)=-\frac1n\sum_{i=1}^n x_i.
\end{equation}
This is the empirical expected-shortfall formula used in
\cite[Section~2]{ACJP}; it includes the fractional boundary weight
and preserves the empirical law under replication.

\begin{lemma}\label{lem:finite}
Let $r:\R^n\to\R$ be law invariant and coherent. Its canonical polar
\[
 D=\{a\in\Delta_n:-a\cdot x\leqslant r(x)\text{ for every }x\in\R^n\}
\]
is non-empty, compact, convex and permutation invariant. With
$K=D\cap\Delta_n^\downarrow$,
\begin{equation}\label{eq:finite-support}
 r(x)=\max_{a\in D}(-a\cdot x)
     =\max_{a\in K}\left(-\sum_{i=1}^na_ix_{i:n}\right).
\end{equation}
For $a\in\Delta_n^\downarrow$, set $a_{n+1}=0$ and
$\mu_k=k(a_k-a_{k+1})$. Then $\mu\in\Delta_n$ and
\begin{equation}\label{eq:discrete-kusuoka}
 -\sum_{i=1}^na_ix_{i:n}
 =\sum_{k=1}^n\mu_k\left(-\frac1k\sum_{i=1}^kx_{i:n}\right).
\end{equation}
\end{lemma}

\begin{proof}
A continuous sublinear functional in finite dimensions is the support
function of the set of linear functionals it dominates, by the
Hahn--Banach theorem. Write a dominated linear functional as $-a\cdot x$.
Testing constants gives $\sum_i a_i=1$. Testing $x=t e_i$ with $t>0$
and using monotonicity gives $-ta_i\leqslant r(te_i)\leqslant0$, so
$a_i\geqslant0$. This proves the first equality and all the asserted
properties of $D$, except permutation invariance, which follows by
permuting the defining inequalities. The rearrangement inequality then
pairs the largest weights with the smallest coordinates and proves the
second equality. Finally, $\mu_k\geqslant0$,
$\sum_k k(a_k-a_{k+1})=\sum_ka_k=1$, and
$\sum_{k=i}^n\mu_k/k=a_i$. Reversing the two finite sums gives
\eqref{eq:discrete-kusuoka}.
\end{proof}

This is the finite spectral, or discrete Kusuoka, representation
\cite{ACJP,NoyanRudolf,Ugurlu}. In particular, the step spectrum
\begin{equation}\label{eq:step-spectrum}
 \phi_{m,a}=n_m\sum_{i=1}^{n_m}a_i\ind_{I_{i,m}},
 \qquad a\in\Delta_{n_m}^\downarrow,
\end{equation}
defines $\rho_{\phi_{m,a}}=\sum_{k=1}^{n_m}\mu_k\ES_{k/n_m}$.
No continuity or strict decrease of a spectrum is required.

We also recall the comonotonic $L$-estimator characterisation from
\cite[Theorem~3.10]{ACJP}, including its weight-recovery formula.

\begin{proposition}\label{prop:comonotone-cre}
A law-invariant CRE $R_n$ is comonotonically additive if and only if
$R_n=L_w$ for a unique $w\in\Delta_n^\downarrow$. Its weights satisfy
\[
 w_i=R_n(-\ind_{\{1,\ldots,i\}})
       -R_n(-\ind_{\{1,\ldots,i-1\}}),\qquad 1\leqslant i\leqslant n.
\]
\end{proposition}

\begin{proof}
For $w\in\Delta_n^\downarrow$, \eqref{eq:discrete-kusuoka} proves
coherence. Comonotone vectors admit a common increasing ordering, so
their order-statistic sums are additive. Conversely, set
$v^k=-\ind_{\{1,\ldots,k\}}$ and $h_k=R_n(v^k)$, with $h_0=0$.
Monotonicity and cash additivity give
$0=h_0\leqslant h_1\leqslant\cdots\leqslant h_n=1$.
For $1\leqslant k<n$, two negative indicators of $k$-element sets with
intersection of size $k-1$ have a sum equidistributed with
$v^{k-1}+v^{k+1}$. Subadditivity, law invariance and comonotonic
additivity therefore give $h_{k-1}+h_{k+1}\leqslant2h_k$.
Thus $w_i=h_i-h_{i-1}$ belongs to $\Delta_n^\downarrow$.
For an increasingly ordered $x$,
\[
 x=x_n\ind+\sum_{k=1}^{n-1}(x_{k+1}-x_k)v^k.
\]
All the summands are comonotone. Applying the axioms and summation by
parts yields $R_n(x)=-\sum_iw_i x_i$. Law invariance gives the formula
for arbitrary $x$, and the displayed differences determine $w$ uniquely.
\end{proof}

Lemma~\ref{lem:finite} expresses a general
law-invariant CRE as a maximum over these estimators; it need not itself
be comonotonically additive.

\subsection{Conditional averages}

The following standard form of dilatation monotonicity is included to make
the approximation and uniqueness arguments explicit.

\begin{lemma}\label{lem:dilatation}
Let $\eta:\Linf[0,1]\to\R$ be law invariant, convex and have the Fatou
property. Then $\eta(\E_mX)\leqslant\eta(X)$ for every $m$ and $X$.
Consequently, if $\eta$ is also coherent,
\begin{equation}\label{eq:dyadic-limit}
 \eta(\E_mX)\uparrow\eta(X).
\end{equation}
\end{lemma}

\begin{proof}
Fix $m$. Within each cell $I_{j,m}$ identify the relative coordinate with
$[0,1)$, and translate it modulo one by $k/2^\ell$, for
$k=0,\ldots,2^\ell-1$. These transformations $T_k$ preserve Lebesgue
measure. The average
\[
 V_\ell X=2^{-\ell}\sum_{k=0}^{2^\ell-1}X\circ T_k
\]
therefore satisfies $\eta(V_\ell X)\leqslant\eta(X)$ by convexity and law
invariance. If $X$ is constant on a dyadic refinement of level $m+L$,
then $V_\ell X=\E_mX$ whenever $\ell\geqslant L$. Since $V_\ell$ and
$\E_m$ are contractions on $L^1$, density of dyadic step functions
gives $V_\ell X\to\E_mX$ in $L^1$ for every bounded $X$. Select an
almost-surely convergent subsequence. Its terms are bounded by
$\|X\|_\infty$, so the Fatou property gives $\eta(\E_mX)\leqslant\eta(X)$.

Applying this inequality to $\E_{m+1}X$ proves monotonicity in
\eqref{eq:dyadic-limit}, since $\E_m\E_{m+1}X=\E_mX$.
The dyadic martingale theorem gives $\E_mX\to X$ almost surely, with
$\|\E_mX\|_\infty\leqslant\|X\|_\infty$. The Fatou property supplies
the reverse limiting inequality and proves the claim.
\end{proof}

\subsection{Nonstandard conventions}\label{subsec:nsa-prelim}

We fix a countably saturated nonstandard enlargement; see
\cite{AFHL,Robinson}. A hyperreal is limited if its absolute value is
bounded by a standard real number. Each limited hyperreal $z$ has a
unique standard part $\st z\in\R$ with $z-\st z$ infinitesimal.
We write $z\simeq w$ for infinitesimal difference. The star transform
of a standard object is denoted by ${}^*$.
An internal set is hyperfinite if it admits an internal bijection with
$\{1,\ldots,H\}$ for some hyperinteger $H$; an internal set need not
be hyperfinite. Hyperfinite sums and maxima over non-empty hyperfinite
sets are provided by transfer.

For $H\in{}^*\N$, the notation ${}^*\R^H$ denotes the internal vector
space of internal real vectors of length $H$. An internal CRE satisfies
the four sample-map axioms with internal vectors and hyperreal scalars.
Internal law invariance uses all internal permutations. Internal
comonotonic additivity is defined by the transferred pairwise inequalities.

An internal probability vector $\mu=(\mu_k)_{k=1}^H$ gives
$p_\mu(A)=\sum_{k\in A}\mu_k$ on internal subsets of the hyperfinite
index set. Its Loeb measure $L(p_\mu)$ is the completed countably additive
measure extending $A\mapsto\st p_\mu(A)$ \cite{Loeb1975}.
For an internal function $f$ with a uniform standard bound,
\[
 \int \st f\,dL(p_\mu)=\st\left(\sum_{k=1}^H\mu_kf(k)\right).
\]
This bounded integration identity is the only Loeb integration result
needed below. We use a Loeb space of tail levels, without identifying
it with the original Lebesgue space of positions.

\section{Reconstruction and canonical dual sets}\label{sec:reconstruction}

\subsection{Compatibility of finite dual sets}

Let $D_m$ be the polars in Theorem~\ref{thm:reconstruction}. In addition
to duplication, use the maps
\[
 (\agg_m b)_j=b_{2j-1}+b_{2j},\qquad
 \spl_m a=(a_1/2,a_1/2,\ldots,a_{n_m}/2,a_{n_m}/2).
\]
Thus $\agg_m\spl_m$ is the identity. The aggregation map acts on
probability weights; the duplication map acts on positions.

\begin{lemma}\label{lem:projective}
Under \eqref{eq:replication},
\begin{equation}\label{eq:projective}
 \agg_m(D_{m+1})=D_m,\qquad \spl_m(D_m)\subseteq D_{m+1}.
\end{equation}
With $K_m=D_m\cap\Delta_{n_m}^\downarrow$, the same assertions hold
for $K_m$ in place of $D_m$.
\end{lemma}

\begin{proof}
For every $x\in\R^{n_m}$,
\[
 \max_{b\in D_{m+1}}(-\agg_m b\cdot x)
 =r_{m+1}(\rep_mx)=r_m(x)=\max_{a\in D_m}(-a\cdot x).
\]
Equality of support functions of compact convex sets gives the first
identity. Given $a\in D_m$, choose $b\in D_{m+1}$ with $\agg_mb=a$.
Average $b$ over the finite group that independently interchanges the
two coordinates in each adjacent pair. Convexity and permutation
invariance of $D_{m+1}$ imply that this average, which is $\spl_ma$,
belongs to $D_{m+1}$. Both aggregation and splitting preserve decreasing
order. Thus $\agg_m(K_{m+1})\subseteq K_m$, while splitting provides a
preimage in $K_{m+1}$ of every $a\in K_m$.
\end{proof}

\subsection{Existence and uniqueness}

\begin{proof}[Proof of Theorem~\ref{thm:reconstruction}: the extension]
Necessity follows from $S_{m+1}\rep_mx=S_mx$. For sufficiency define
\begin{equation}\label{eq:spectral-envelope}
 \rho(X)=\sup_{m\geqslant0}\ \max_{a\in K_m}
             \left(-\int_0^1q_X(u)\phi_{m,a}(u)\,du\right).
\end{equation}
The inner maximum exists by compactness of $K_m$. Each displayed spectral
functional is a finite mixture of expected shortfalls, so it is coherent,
law invariant and has the Lebesgue property. Its absolute value is at most
$\|X\|_\infty$. Their supremum is therefore a finite coherent
law-invariant risk measure and has the Fatou property.

We check its restriction at level $\ell$. For $X=S_\ell x$,
$q_X$ consists of the ordered coordinates $x_{i:n_\ell}$ on consecutive
equal cells. If $m\geqslant\ell$, the integral against $\phi_{m,a}$
groups $2^{m-\ell}$ consecutive weights together. By iterating
Lemma~\ref{lem:projective}, the grouped vector belongs to $K_\ell$,
so that integral with its minus sign is at most $r_\ell(x)$.
If $m<\ell$, split each weight uniformly to level $\ell$; the resulting
vector belongs to $K_\ell$ and gives exactly the same integral. At
$m=\ell$, maximisation gives $r_\ell(x)$ by
Lemma~\ref{lem:finite}. This proves \eqref{eq:restrictions}.

Lemma~\ref{lem:dilatation} now gives
$\rho(X)=\lim_m\rho(\E_mX)=\lim_m r_m(P_mX)$ and proves
\eqref{eq:reconstruction}, including monotonicity. If another law-invariant
coherent risk has the same restrictions, it has the Fatou property by
\cite[Theorem~2.2]{JST}, and the same lemma reconstructs it from exactly
this sequence. The extension is therefore unique in the asserted class.
\end{proof}

\subsection{Recovery of the population polar}

\begin{proof}[Proof of Theorem~\ref{thm:reconstruction}: the dual assertions]
Splitting a weight vector does not change its step density. Consequently
Lemma~\ref{lem:projective} gives $\mathcal C_m\subseteq\mathcal C_{m+1}$.
For $Z=n_m\sum_j a_j\ind_{I_{j,m}}\in\mathcal C_m$ and any $X\in\Linf$,
\[
 \E[-XZ]=-a\cdot P_mX\leqslant r_m(P_mX)\leqslant\rho(X).
\]
Thus $\mathcal C_m\subseteq\Zrho$, and norm closedness of $\Zrho$
gives one inclusion in \eqref{eq:dual-closure}.

Conversely, for $Z\in\Zrho$, testing \eqref{eq:canonical-duality}
against $S_mx$ gives
\[
 -\sum_{j=1}^{n_m}x_j\int_{I_{j,m}}Z\,d\lambda\leqslant r_m(x)
 \qquad(x\in\R^{n_m}).
\]
Hence its dyadic mass vector belongs to $D_m$, and $\E_mZ\in\mathcal C_m$.
Since $\E_mZ\to Z$ in $L^1$, this proves the other inclusion.
The same test gives containment of $\Zrho$ in the set in
\eqref{eq:dual-moments}. If a density satisfies all the moment constraints,
then $\E_mZ\in\mathcal C_m\subseteq\Zrho$ for every $m$; norm convergence
and closedness put $Z$ in $\Zrho$. This proves \eqref{eq:dual-moments}.
Finally, every $a\in D_m$ is the dyadic mass vector of its density in
$\mathcal C_m$. This gives the reverse inclusion in
\eqref{eq:dual-projection} and completes the proof.
\end{proof}

\begin{remark}
No compactness of $\Zrho$ is needed for reconstruction. In particular,
$r_m(x)=-\min_jx_j$ satisfies replication consistency, and its extension
is $\rho(X)=-\operatorname{ess\,inf}X$. Here $D_m=\Delta_{n_m}$ and
$\Zrho$ is the set of all non-negative $L^1$ densities of mass one.
That set is not uniformly integrable.
\end{remark}

\begin{example}\label{ex:incompatible}
The distinction is visible in the expected-shortfall estimators compared
in \cite[Example~2.3 and Section~5]{ACJP}. For a fixed tail probability $\alpha$, the
sample-tail mean uses $k_n=\max\{1,\lfloor n\alpha\rfloor\}$ and
\[
 T_n(x)=-\frac1{k_n}\sum_{i=1}^{k_n}x_{i:n}.
\]
The maximum with one extends the usual formula to samples with
$\lfloor n\alpha\rfloor=0$. Each $T_n$ is a law-invariant comonotone CRE.
For $\alpha=2/5$ and
$x=(-3,-1,0,2)$, however,
$T_4(x)=3$ whereas $T_8(\mathsf D_2x)=7/3$.
By contrast, empirical expected shortfall with the fractional boundary
weight in \eqref{eq:fractional-empirical-es} equals $9/4$ on both
samples. The fixed-size CRE axioms allow both procedures; replication
consistency distinguishes them.
\end{example}

\section{The exact modulus and finite stability tests}\label{sec:modulus}

\subsection{Computing the Bernoulli profile}

In this section write $c=c_\rho$. By law invariance and canonical duality,
\begin{equation}\label{eq:profile-dual}
 c(t)=\sup_{Z\in\Zrho}\ \sup_{\lambda(A)=t}\int_AZ\,d\lambda.
\end{equation}
It follows that $c$ is non-decreasing, $c(0)=0$, $c(1)=1$, and
$0\leqslant c(t)\leqslant1$. To justify the double supremum, note that
for every fixed $A$ of measure $t$, duality gives
$\sup_Z\int_AZ=\rho(-\ind_A)=c(t)$.

\begin{proposition}\label{prop:finite-profile}
For $a\in K_m$ and $0\leqslant t<1$, set $k=\lfloor n_mt\rfloor$ and
\[
 L_a(t)=\sum_{i=1}^ka_i+(n_mt-k)a_{k+1},\qquad L_a(1)=1.
\]
Then
\begin{equation}\label{eq:finite-profile}
 c(t)=\sup_{m\geqslant0}\max_{a\in K_m}L_a(t).
\end{equation}
At dyadic probabilities the expression is finite: for $0\leqslant k\leqslant n_m$,
\begin{equation}\label{eq:dyadic-profile}
 c(k/n_m)=r_m(-\ind_{\{1,\ldots,k\}})
         =\max_{a\in K_m}\sum_{i=1}^ka_i.
\end{equation}
\end{proposition}

\begin{proof}
For a step density with decreasing masses $a$, a set of measure $t$
maximises its integral by taking the first $k$ cells and the appropriate
fraction of the next cell. This gives $L_a(t)$; rearrangement reduces
an unsorted density to this case. The map
$Z\mapsto\sup_{\lambda(A)=t}\int_AZ$ is $1$-Lipschitz in the $L^1$ norm,
since the difference of its values at $Z,W$ is at most $\|Z-W\|_1$.
Equations~\eqref{eq:dual-closure} and \eqref{eq:profile-dual} therefore
give \eqref{eq:finite-profile}. For dyadic $t$, the indicator is already
a position on the level-$m$ model; \eqref{eq:restrictions} and
\eqref{eq:finite-support} give \eqref{eq:dyadic-profile}.
\end{proof}

\begin{remark}\label{rem:nonconcave}
Although each $L_a$ is concave, $c$ need not be. For example, take
\[
 \rho=\max\left\{\ES_{1/2},\ \tfrac15\ES_{1/10}
                                  +\tfrac45(-\E)\right\}.
\]
Its profile is
$c(t)=\max\{\min(2t,1),\ \frac15\min(10t,1)+\frac45t\}$.
At $t=1/6$ the slope increases from $4/5$ to $2$, so $c$ is not concave.
The proof of the exact modulus below does not apply Jensen's inequality
to $c$.
\end{remark}

\subsection{Proof of the modulus formula}

\begin{proof}[Proof of Theorem~\ref{thm:modulus}: formula~\eqref{eq:exact-modulus}]
Realise $P,Q\in\Pa$ by their quantiles $X=q_P$, $Y=q_Q$ on $[0,1]$.
Then $h=|X-Y|$ satisfies $0\leqslant h\leqslant B$ and
$\E h=\Wone(P,Q)$. Subadditivity and monotonicity give
\begin{equation}\label{eq:modulus-first}
 |\rho(X)-\rho(Y)|\leqslant\rho(-h)
                  =\sup_{Z\in\Zrho}\E[hZ].
\end{equation}
For example, $\rho(X)-\rho(Y)\leqslant\rho(X-Y)\leqslant\rho(-h)$,
and the other direction is identical.

Set $t=\E h/B$ and first suppose $0<t<1$. For fixed $Z\in\Zrho$,
choose a threshold $z\geqslant0$ and a measurable set $A$ such that
\[
 \lambda(A)=t,\qquad \{Z>z\}\subseteq A\subseteq\{Z\geqslant z\}.
\]
The quantile of $Z$ provides such a threshold, and atomlessness permits
splitting the level set $\{Z=z\}$. Since $0\leqslant h/B\leqslant1$,
\[
 (h/B-\ind_A)(Z-z)\leqslant0
 \quad\text{almost surely},\qquad \E(h/B-\ind_A)=0.
\]
Integration gives $\E[hZ]\leqslant B\int_AZ\leqslant Bc(t)$.
If $t=0$, then $h=0$ almost surely; if $t=1$, the bound follows from
$h\leqslant B$ and $\E Z=1$. Taking the supremum over $Z$ in
\eqref{eq:modulus-first} proves
\begin{equation}\label{eq:modulus-bound}
 |\rho(P)-\rho(Q)|\leqslant Bc\bigl(\Wone(P,Q)/B\bigr).
\end{equation}
Monotonicity of $c$ gives the upper bound in \eqref{eq:exact-modulus}.

For equality, let $t=e/B$ and take $Q=\delta_{\overline x}$ and
$P=t\delta_{\underline x}+(1-t)\delta_{\overline x}$. Their distance is $Bt=e$. Realise the
second law as $\overline x-B\ind_A$ with $\lambda(A)=t$. Cash additivity and
positive homogeneity give
$\rho(P)=-\overline x+Bc(t)$ and $\rho(Q)=-\overline x$. Their difference is the required
$Bc(e/B)$, including both endpoints $e=0,B$.
\end{proof}

\begin{corollary}\label{cor:other-moduli}
Let $d_{\rm K}(P,Q)=\sup_x|F_P(x)-F_Q(x)|$ and
$d_{\rm TV}(P,Q)=\sup_{A\in\mathcal B(\R)}|P(A)-Q(A)|$.
For $d\in\{d_{\rm K},d_{\rm TV}\}$ and $0\leqslant t\leqslant1$,
\[
 \sup_{\substack{P,Q\in\Pa\\d(P,Q)\leqslant t}}
       |\rho(P)-\rho(Q)|=B c_\rho(t).
\]
\end{corollary}

\begin{proof}
On the common support interval,
$\Wone(P,Q)\leqslant B d_{\rm K}(P,Q)\leqslant B d_{\rm TV}(P,Q)$.
Theorem~\ref{thm:modulus} gives the upper bounds. The laws
$Q=\delta_{\overline x}$ and $P=t\delta_{\underline x}+(1-t)\delta_{\overline x}$ have both displayed
distances equal to $t$ and risk difference $Bc_\rho(t)$.
\end{proof}

\subsection{Continuity and the response to one loss}

\begin{proof}[Proof of Theorem~\ref{thm:modulus}: the equivalences]
Equation~\eqref{eq:dyadic-profile} with $k=1$ gives $d_m=c(2^{-m})$.
Monotonicity of $c$ proves the equivalence of (1) and (2).
Equation~\eqref{eq:profile-dual}, together with uniform absolute
continuity of integrals of non-negative densities of mass one, proves
the equivalence of (2) and (3). The canonical dual set is convex and
norm closed, hence weakly closed, and is bounded in $L^1$. The
Dunford--Pettis criterion therefore gives the equivalence of (3) and (4).

To prove (3)$\Rightarrow$(5), suppose $X_k\to X$ almost surely and
$\|X_k\|_\infty,\|X\|_\infty\leqslant M$. For each $\eta>0$,
\[
 \sup_{Z\in\Zrho}\E[|X_k-X|Z]
 \leqslant\eta+2M\sup_{Z\in\Zrho}
                    \int_{\{|X_k-X|>\eta\}}Z\,d\lambda.
\]
The measures of these sets tend to zero. Uniform absolute continuity
and then $\eta\downarrow0$ show that the left side tends to zero.
Duality bounds $|\rho(X_k)-\rho(X)|$ by this quantity. Conversely,
applying the Lebesgue property to $-\ind_{(0,2^{-m}]}\to0$ proves (1).
Finally, the exact modulus \eqref{eq:exact-modulus} makes (2) and (6)
equivalent.
\end{proof}

\begin{corollary}\label{cor:certificate}
For $P,Q\in\Pa$ and $m\geqslant0$,
\begin{equation}\label{eq:finite-certificate}
 \Wone(P,Q)\leqslant B2^{-m}
 \quad\Longrightarrow\quad
 |\rho(P)-\rho(Q)|\leqslant B r_m(-e_1).
\end{equation}
If $X$ takes values in $[\underline x,\overline x]$, then
\begin{equation}\label{eq:conditional-error}
 0\leqslant\rho(X)-r_m(P_mX)
 \leqslant Bc\bigl(\|X-\E_mX\|_1/B\bigr).
\end{equation}
\end{corollary}

\begin{proof}
The first assertion follows from \eqref{eq:modulus-bound}, monotonicity
of $c$ and $c(2^{-m})=r_m(-e_1)$. For the second,
$\E_mX$ has the same range bound, its law is at $\Wone$ distance at
most $\|X-\E_mX\|_1$ from the law of $X$, and
$\rho(\E_mX)=r_m(P_mX)\leqslant\rho(X)$.
\end{proof}

\begin{remark}
The convergence in Theorem~\ref{thm:reconstruction} holds even when
$c(0+)>0$. In that case \eqref{eq:conditional-error} need not give a
vanishing upper bound. Conditional approximation of each fixed position
and uniform continuity over all bounded laws are different assertions.
\end{remark}

\section{Uniform estimation on bounded supports}\label{sec:statistics}

An estimator with auxiliary randomisation is a Borel measurable function
$T_n:[\underline x,\overline x]^n\times[0,1]\to\R$, where the last coordinate is independent
and uniform. Deterministic estimators ignore that coordinate. This
convention suffices to represent real-valued randomised procedures on
these standard Borel spaces. Uniform consistency means
\[
 \lim_{n\to\infty}\sup_{P\in\Pa}
 (P^{\otimes n}\otimes\lambda)
     (|T_n-\rho(P)|>\varepsilon)=0
 \qquad\text{for every }\varepsilon>0.
\]
The empirical plug-in map is measurable: assigning a vector to the
equal-cell step function on $n$ cells and applying $\rho$ gives a
$1$-Lipschitz function of that vector in the maximum norm, by
\eqref{eq:supnorm-lipschitz}. For $n=n_m$ it is exactly $r_m$.
The construction thus supplies coherent empirical functionals at every
sample size, including the non-dyadic ones.

\begin{proof}[Proof of Theorem~\ref{thm:statistics}: the upper bound]
For $P\in\Pa$, let $F$ and $F_n$ be its distribution function and its
empirical distribution function. Equation~\eqref{eq:wasserstein} gives
\[
 \Wone(\widehat P_n,P)\leqslant B\|F_n-F\|_\infty.
\]
The Dvoretzky--Kiefer--Wolfowitz inequality with Massart's constant
\cite{Massart} states that, for every $u>0$,
\[
 \sup_{P\in\Pa}P^{\otimes n}(\|F_n-F\|_\infty>u)
 \leqslant2\exp(-2nu^2).
\]
If $u_{n,\delta}<1$, combine this with \eqref{eq:modulus-bound} and
monotonicity of $c$ to obtain \eqref{eq:concentration}. If
$u_{n,\delta}=1$, the risk error is at most $B=Bc(1)$, so the
probability in \eqref{eq:concentration} is zero.
When $c(0+)=0$, for any fixed $\delta>0$ the bound
$Bc(u_{n,\delta})$ tends to zero. Given $\varepsilon>0$, first choose
$\delta$ arbitrarily small and then $n$ sufficiently large. This proves
uniform consistency of the plug-in estimators, and hence existence of
a uniformly consistent sequence.
\end{proof}

\begin{proof}[Proof of Theorem~\ref{thm:statistics}: necessity]
Fix $0<t<1$ and compare $Q=\delta_{\overline x}$ with
$P_t=t\delta_{\underline x}+(1-t)\delta_{\overline x}$. Their risks differ by $Bc(t)$.
Under both laws the observation $(\overline x,\ldots,\overline x)$ is possible, with
probabilities $1$ and $s=(1-t)^n$, respectively. On that observation,
let $p_0,p_1$ denote the conditional probabilities that the estimator
errs by more than $\varepsilon$ for the respective target values.
If $2\varepsilon<Bc(t)$, the two closed success intervals are disjoint.
The same auxiliary randomisation is used to assess both targets, so
$p_0+p_1\geqslant1$. The two unconditional error probabilities are
$p_0$ and at least $s p_1$. Consequently
\[
 \sup_{P\in\Pa}\Pbb_P(|T_n-\rho(P)|>\varepsilon)
 \geqslant\max\{p_0,s p_1\}\geqslant\frac{s}{1+s},
\]
which is \eqref{eq:testing-lower}. If $c_0=c(0+)>0$, choose
$0<\varepsilon<Bc_0/2$. Letting $t\downarrow0$ gives a lower bound of
$1/2$ for every $n$ and every estimator. Uniform consistency is
impossible. Together with the upper bound and
Theorem~\ref{thm:modulus}, this proves all the equivalences.
\end{proof}

\subsection{Consequences and limits of the bounds}

The finite test gives a directly evaluable confidence radius. If
$n\geqslant 2^{2m-1}\log(2/\delta)$, then
$u_{n,\delta}\leqslant2^{-m}$ and
\begin{equation}\label{eq:finite-confidence}
 \sup_{P\in\Pa}P^{\otimes n}
 \bigl(|\widehat\rho_n-\rho(P)|>B r_m(-e_1)\bigr)\leqslant\delta.
\end{equation}
This is a sufficient sample-size bound. It is not asserted to be
necessary or rate optimal.

\begin{proposition}\label{prop:bounded-spectra}
For laws with finite first moment, suppose
$\rho(P)=\sup_{\phi\in\Phi}\rho_\phi(P)$, where $\Phi$ is a
non-empty family of spectra satisfying
$\sup_{\phi\in\Phi}\|\phi\|_\infty\leqslant C<\infty$. For any two such laws,
\[
 |\rho(P)-\rho(Q)|\leqslant C\Wone(P,Q),\qquad
 c(t)\leqslant\min\{Ct,1\}.
\]
For every sample $x\in\R^n$ and every $\phi\in\Phi$,
\begin{equation}\label{eq:empirical-exact}
 \rho_\phi\left(\frac1n\sum_{i=1}^n\delta_{x_i}\right)
 =-\sum_{i=1}^nx_{i:n}\int_{(i-1)/n}^{i/n}\phi(u)\,du.
\end{equation}
\end{proposition}

\begin{proof}
The integral defining a spectral risk differs at two laws by at most
$\|\phi\|_\infty\int_0^1|q_P-q_Q|$. Taking the supremum over
$\phi$ proves the first bound. For a Bernoulli loss,
$c(t)=\sup_{\phi\in\Phi}\int_0^t\phi(u)\,du$, which proves the
second. The empirical quantile equals $x_{i:n}$ on the $i$th cell,
so splitting its integral gives \eqref{eq:empirical-exact}.
\end{proof}

The identity \eqref{eq:empirical-exact} has no discretisation remainder
and does not require a Lipschitz spectrum. It is a standard quantile
calculation; the concentration bound then recovers familiar bounded-spectrum
estimates. Theorem~\ref{thm:statistics} also covers risks outside this
class. For instance, $\phi(u)=\gamma u^{\gamma-1}$,
$0<\gamma<1$, has $c(t)=t^\gamma$ and an unbounded spectrum. Its
distribution-free radius from \eqref{eq:concentration} has order
$n^{-\gamma/2}$ at fixed confidence.

There is generally a gap between the statistical upper and lower bounds:
the former evaluates $c$ at order $n^{-1/2}$, whereas choosing $t$ of
order $n^{-1}$ in \eqref{eq:testing-lower} gives a non-vanishing lower
error probability at an error scale of order $c(1/n)$. No equality of
these rates is claimed. Moreover, uniform estimation concerns all laws
on a fixed support interval. Pointwise consistency for each fixed law
holds for every fixed bounded law, even when the uniform criterion fails;
see Lemma~\ref{lem:uniform-es-consistency} and its ensuing discussion.
The common support bound is an essential hypothesis here; no
extension of the distribution-free statements to arbitrary unbounded
laws is implicit.

\section{Nonstandard constructions of coherent estimators}\label{sec:nsa-construction}

The sample-map representations of \cite{ACJP} describe admissible weights
at a fixed sample size. Here one internal estimator generates standard
sample maps simultaneously, and Loeb integration identifies their weights.

\subsection{One hyperfinite estimator and all finite sample sizes}

\begin{proof}[Proof of Theorem~\ref{thm:shadow}]
Every standard positive integer divides $H$, by transfer of
the corresponding property of factorials.  The repetition map extends to an
internal linear map from ${}^*\R^n$ into ${}^*\R^H$.  It preserves the
coordinate order and constants.  The transferred monotonicity
and cash-additivity axioms imply
$|\mathfrak r(\iota_{n,H}x)|\leqslant\|x\|_\infty$, so its standard
part exists.  Applying the four internal CRE axioms to
$\iota_{n,H}x,\iota_{n,H}y$, and then taking standard parts, proves the four
ordinary CRE axioms for $R_n$.  A permutation of the coordinates
of $x$ induces an internal permutation of the blocks of $\iota_{n,H}x$;
hence $R_n$ is law invariant.  As before, monotonicity and cash
additivity give the Lipschitz bound
$|R_n(x)-R_n(y)|\leqslant\|x-y\|_\infty$, in particular
measurability.  Finally, $\iota_{rn,H}(\mathsf D_rx)$ is a permutation of
$\iota_{n,H}x$: both contain exactly $H/n$ copies of each coordinate
$x_i$.  Internal law invariance gives the claimed replication
identity before taking standard parts.  The dyadic reconstruction
theorem gives a unique law-invariant Fatou extension $\rho$ of
$(R_{2^m})_{m\geqslant0}$.

We verify its agreement with every $R_n$.  Fix $n$ and $x$, and
let $d=2^m$.  Write $\ell_i$ and $u_i$ for the essential infimum
and supremum of the empirical quantile of $x$ on the $i$th
dyadic cell.  Comparison at the common dimension $nd$, using
replication and permutation invariance, gives
$R_d(\ell)\geqslant R_n(x)\geqslant R_d(u)$.
The same bounds hold for $\rho(P_x)$ by monotonicity.
The nondecreasing vectors $\ell,u$ differ in at most $n-1$
coordinates, all with index $i>d/n$.  Every non-increasing
probability vector $a$ satisfies $a_i\leqslant1/i$.
Choose a decreasing probability vector $a^*$ attaining the maximum for
$R_d(\ell)$ in \eqref{eq:finite-support}. Since
$R_d(u)\geqslant-a^*\cdot u$, we have
$R_d(\ell)-R_d(u)\leqslant a^*\cdot(u-\ell)$. The bound
$a_i^*\leqslant1/i$ therefore yields
\[
 0\leqslant R_d(\ell)-R_d(u)
    \leqslant\frac{n(n-1)}d
                   \bigl(\max_i x_i-\min_i x_i\bigr)
    \longrightarrow0.
\]
The two quantities trapped between these bounds are equal.
Uniqueness already follows from the dyadic restrictions.

If $x,y$ are comonotone, their repetitions are internally comonotone,
so internal additivity passes to standard parts. Transfer of
Proposition~\ref{prop:comonotone-cre} supplies a decreasing internal
probability vector $(a_j)$ for an internally comonotone $\mathfrak r$.
Grouping its weights over the $H/n$ repeated order statistics gives
\eqref{eq:shadow-weights}; the grouped standard weights remain a
non-increasing probability vector.
\end{proof}

\begin{proposition}\label{prop:all-sample-extension}
A family $(R_n)_{n\geqslant1}$ of law-invariant CREs satisfies
\eqref{eq:all-replication} if and only if there is a unique
law-invariant coherent risk measure $\rho$ on $L^\infty[0,1]$
with $R_n(x)=\rho(P_x)$ for every $n$ and $x\in\R^n$.
Every such family is obtained by the construction in
Theorem~\ref{thm:shadow}.
\end{proposition}

\begin{proof}
Replication is necessary because $P_{\mathsf D_kx}=P_x$.
For sufficiency, apply Theorem~\ref{thm:reconstruction} to the dyadic
subfamily. The trapping argument above uses only replication and
law invariance, and therefore identifies its extension with every $R_n$.
To prove the final assertion, regard
$\mathcal R(n,x)=\rho(P_x)$ as a standard function on
$\bigcup_{n\geqslant1}(\{n\}\times\R^n)$, and put
$\mathfrak r(z)=({}^*\mathcal R)(H,z)$ for $H=J!$.
Transfer gives the internal CRE axioms and permutation invariance.
Transferring the replication identity with factor $H/n$ yields
$\mathfrak r(\iota_{n,H}x)=R_n(x)$ for every standard $n$ and $x$;
in particular, taking standard parts recovers the prescribed family.
\end{proof}

\subsection{Loeb mixtures of empirical expected shortfalls}

The restriction $H=J!$ permits exact equal-size blocks for every
standard sample size.  It is used only in
Theorem~\ref{thm:shadow}.  For the explicit
construction below, $H$ may be any infinite hyperinteger.
If
$\mu=(\mu_k)_{k=1}^H$ is an internal probability vector, its internal
probability on $\{1,\ldots,H\}$ is
$p_\mu(A)=\sum_{k\in A}\mu_k$.  We write $L(p_\mu)$ for its Loeb
probability measure \cite{Loeb1975}, and set
\begin{equation}\label{eq:loeb-pushforward}
 \nu_\mu=\bigl(k\mapsto\st(k/H)\bigr)_*L(p_\mu).
\end{equation}
Thus $\nu_\mu$ is an ordinary Borel probability measure on $[0,1]$.
In particular, mass carried by infinitesimal positive levels may produce
an atom at zero.

\begin{theorem}\label{thm:loeb-cre-construction}
Let $\mathscr M$ be a non-empty internal hyperfinite set of internal
probability vectors of length $H$.  For each standard $n\geqslant1$,
define
\begin{equation}\label{eq:internal-estimator-shadow}
 R_n(x)=\st\left(
       \max_{\mu\in\mathscr M}
       \sum_{k=1}^H\mu_k\,
       {}^*\!\bigl[\alpha\mapsto
                    \ES_\alpha(P_x)\bigr](k/H)
                              \right),
 \qquad x\in\R^n.
\end{equation}
Let $\Gamma$ be the weak closure in $\mathcal P([0,1])$ of
$\{\nu_\mu:\mu\in\mathscr M\}$.  Then $\Gamma$ is non-empty and
compact, and
\begin{equation}\label{eq:loeb-standard-estimator}
 R_n(x)=\max_{\nu\in\Gamma}
          \int_{[0,1]}\ES_\alpha(P_x)\,\nu(d\alpha).
\end{equation}
Each $R_n$ is a law-invariant CRE as defined in \cite{ACJP}.
Replication holds at every sample size:
\begin{equation}\label{eq:loeb-replication}
 R_{rn}(\mathsf D_r x)=R_n(x),\qquad n,r\in\N,
\end{equation}
where $\mathsf D_r x$ repeats each coordinate $r$ times.  Moreover,
\begin{equation}\label{eq:loeb-orderweights}
 R_n(x)=\max_{\nu\in\Gamma}
                \left(-\sum_{i=1}^n w_{i,n}(\nu)x_{i:n}\right),
 \qquad
 w_{i,n}(\nu)=\int_{[0,1]}h_{i,n}(\alpha)\,\nu(d\alpha),
\end{equation}
where
\[
 h_{i,n}(\alpha)=
 \frac{\min\{\alpha,i/n\}-\min\{\alpha,(i-1)/n\}}{\alpha}
 \quad(\alpha>0),
 \qquad
 h_{i,n}(0)=\mathbf1_{\{i=1\}}.
\]
For every $\nu$, these weights satisfy
$w_{1,n}(\nu)\geqslant\cdots\geqslant w_{n,n}(\nu)\geqslant0$
and $\sum_iw_{i,n}(\nu)=1$.
\end{theorem}

\begin{proof}
For fixed standard $x$, the function
$f_x(\alpha)=\ES_\alpha(P_x)$ is continuous on
$[0,1]$ and bounded by $\|x\|_\infty$ in absolute value.  Continuity
at zero follows because $f_x(\alpha)=-x_{1:n}$ when
$0\leqslant\alpha\leqslant1/n$.  Consequently every internal
average in \eqref{eq:internal-estimator-shadow} is limited.  The
Loeb integration theorem gives
\begin{equation}\label{eq:loeb-continuous-test}
 \st\sum_{k=1}^H\mu_k\,{}^*f_x(k/H)
      =\int_{[0,1]}f_x(\alpha)\,\nu_\mu(d\alpha).
\end{equation}
For completeness, this identity follows first for finite step
approximations to $f_x\circ\st$ and then for $f_x$
by uniform continuity.  The map $k\mapsto\st(k/H)$
is Loeb measurable, as is seen by expressing inverse images of open
intervals through countable unions of internal intervals.

The maximum in \eqref{eq:internal-estimator-shadow} is attained,
since $\mathscr M$ is internal and hyperfinite.  Applying
\eqref{eq:loeb-continuous-test} to a maximising vector proves that
its standard part is the maximum of $\int f_x\,d\nu_\mu$ over
$\mu\in\mathscr M$.  The space $\mathcal P([0,1])$ is weakly
compact, and $\nu\mapsto\int f_x\,d\nu$ is weakly continuous.
Taking the closure therefore proves
\eqref{eq:loeb-standard-estimator} and the compactness assertion.

Each empirical expected shortfall in
\eqref{eq:loeb-standard-estimator}, including the level-zero
functional $-\min_i x_i$, is monotone, cash additive, positively
homogeneous and subadditive.  Probability mixtures preserve these
four properties.  Their supremum preserves monotonicity, cash
additivity and positive homogeneity; subadditivity follows from
$\sup_\nu(a_\nu+b_\nu)\leqslant\sup_\nu a_\nu+\sup_\nu b_\nu$.
Permutation invariance is immediate.  These estimators are measurable:
indeed, monotonicity and cash additivity imply
$|R_n(x)-R_n(y)|\leqslant\|x-y\|_\infty$.
Since $P_{\mathsf D_r x}=P_x$, \eqref{eq:loeb-replication} follows.

Integrating the empirical quantile over $[0,\alpha]$ gives
$\ES_\alpha(P_x)=-\sum_i h_{i,n}(\alpha)x_{i:n}$.
The same formula holds at zero.  For each $\alpha$, the vector
$(h_{i,n}(\alpha))_{i=1}^n$ is a non-increasing probability vector.
Integration proves \eqref{eq:loeb-orderweights} and the asserted
properties of the weights.  Notice also that every $h_{i,n}$ is
continuous on $[0,1]$.  Thus the set of weight vectors in
\eqref{eq:loeb-orderweights} is compact.
\end{proof}

Theorem~\ref{thm:loeb-cre-construction} gives precisely the type of
supremum of $L$-estimators characterised in \cite{ACJP}.  For a
singleton $\Gamma$, it gives a comonotone CRE.  A general envelope
need not be comonotone.  The additional conclusions are the
simultaneous construction at all standard sample sizes and the
identification of infinitesimal-level mass in the standard estimator.

The standard formulation does not require an internal input. For any
non-empty weakly compact $\Gamma\subseteq\mathcal P([0,1])$, set
\begin{equation}\label{eq:loeb-population}
 \rho_\Gamma(X)=\max_{\nu\in\Gamma}
       \int_{[0,1]}\ES_\alpha(X)\,\nu(d\alpha),
 \qquad X\in L^\infty,
\end{equation}
where $\ES_0(X)=-\operatorname{ess\,inf}X$.
The function $\alpha\mapsto\ES_\alpha(X)$ is
continuous on $[0,1]$ for bounded $X$, so the maximum is attained.
This is a law-invariant coherent risk measure with the Fatou
property. Its empirical restrictions satisfy
\eqref{eq:loeb-standard-estimator}--\eqref{eq:loeb-orderweights} and
\eqref{eq:loeb-replication} by the same calculation. To check the Fatou
property directly, let $X_j\to X$ almost surely with
$\sup_j\|X_j\|_\infty<\infty$.  For $\alpha>0$, convergence
holds in expected shortfall by its $L^1$ Lipschitz bound
$1/\alpha$; at zero, the required lower-semicontinuity follows
from the essential infimum.  Fatou's lemma applies to each mixture
after adding a common constant, and then to the supremum.

\begin{proposition}\label{prop:loeb-tail-mass}
Let $\Gamma$ be the mixing set constructed from $\mathscr M$ in
Theorem~\ref{thm:loeb-cre-construction}.
For $0<t\leqslant1$, the rare-loss profile of
\eqref{eq:loeb-population} is
\begin{equation}\label{eq:loeb-c-profile}
 c(t)=\max_{\nu\in\Gamma}
     \int_{[0,1]}g_t(\alpha)\,\nu(d\alpha),
 \qquad
 g_t(\alpha)=
 \begin{cases}
  \min\{t/\alpha,1\},&\alpha>0,\\
  1,&\alpha=0.
 \end{cases}
\end{equation}
In particular,
\begin{equation}\label{eq:loeb-zero-mass}
 c(0+)=\max_{\nu\in\Gamma}\nu(\{0\})
 =\lim_{\eta\downarrow0}
     \st\left(
        \max_{\mu\in\mathscr M}
              \sum_{k\leqslant\eta H}\mu_k\right).
\end{equation}
Consequently the uniform continuity and uniform estimability
criteria in the preceding sections hold exactly when the last
quantity is zero.
\end{proposition}

\begin{proof}
The expected shortfall of a loss of size one occurring with
probability $t>0$ is $g_t(\alpha)$.  This proves
\eqref{eq:loeb-c-profile}.  Each $g_t$ is continuous on $[0,1]$,
and $g_t\downarrow\mathbf1_{\{0\}}$ as $t\downarrow0$.
Clearly $c(0+)\geqslant\sup_{\nu\in\Gamma}\nu(\{0\})$.
Conversely, take $t_j\downarrow0$ and maximisers $\nu_j\in\Gamma$.
After passing to a subsequence, $\nu_j$ converges weakly to some
$\nu\in\Gamma$.  For each fixed $s>0$ and all sufficiently large
$j$, $g_{t_j}\leqslant g_s$, whence
\[
 c(0+)\leqslant\lim_j\int g_s\,d\nu_j
             =\int g_s\,d\nu.
\]
Letting $s\downarrow0$ proves
$c(0+)\leqslant\nu(\{0\})$ and attainment of the maximum.

Write
$q_\eta=\st\max_{\mu\in\mathscr M}
                      \sum_{k\leqslant\eta H}\mu_k$.
For $0<t,\eta\leqslant1$, the internal pointwise estimates give
\begin{equation}\label{eq:loeb-threshold-bounds}
 q_\eta\leqslant c(\eta),
 \qquad
 c(t)\leqslant q_\eta+\frac{t}{\eta}.
\end{equation}
The first bound uses $g_\eta(k/H)=1$ when $k\leqslant\eta H$.
For the second, the terms with $k\leqslant\eta H$ contribute
at most their mass, and the remaining terms are bounded by
$t/\eta$.  The standard-part identity for the continuous test
function $g_t$ justifies replacing the internal maximum by $c(t)$.
First let $t\downarrow0$ in the second bound and then let
$\eta\downarrow0$ in both bounds.  This proves the last equality
in \eqref{eq:loeb-zero-mass}, without identifying the mass of
interval endpoints under a Loeb pushforward.
\end{proof}

\subsection{Explicit constructions}\label{subsec:explicit-cre}

The following estimators have ordinary finite formulas.  Hyperfinite
averaging is used to derive the weights and to determine whether
mass survives at level zero; it is not required in their implementation.

First take the single internal vector $\mu_k=1/H$.
Its Loeb pushforward is Lebesgue measure on $[0,1]$.
The resulting population functional is
$\int_0^1\ES_\alpha(X)\,d\alpha$: this is the second-order
superquantile at level zero of the loss $-X$ in
\cite[Section~2]{RockafellarRoyset}. Its empirical estimator is
\begin{equation}\label{eq:logarithmic-cre}
 R_n^{\log}(x)
   =-\sum_{i=1}^n
       \left[
         G\left(\frac in\right)
         -G\left(\frac{i-1}{n}\right)
       \right]x_{i:n},
 \qquad
 G(t)=t(1-\log t),\quad G(0)=0.
\end{equation}
Indeed, Tonelli's theorem identifies its spectrum as
$\int_u^1d\alpha/\alpha=-\log u$, and $G$ is the integral of
this spectrum over $[0,t]$.  In particular,
$w_{1,n}=(1+\log n)/n$.  Although the spectrum is unbounded,
$c(t)=G(t)\to0$ as $t\downarrow0$.  The estimator therefore
satisfies the uniform estimability criterion without a bounded-spectrum
assumption.

More generally, let $\beta>0$ be standard and put
\begin{equation}\label{eq:beta-internal-mixing}
 \mu_k^{(\beta)}
    =\left(\frac{k}{H}\right)^\beta
       -\left(\frac{k-1}{H}\right)^\beta,
 \qquad 1\leqslant k\leqslant H.
\end{equation}
These positive weights sum to one.  Their cumulative sums show
that the Loeb pushforward has distribution function $t^\beta$,
and hence density $\beta\alpha^{\beta-1}$ on $(0,1)$.
The resulting estimator is
\begin{equation}\label{eq:beta-cre}
 R_n^{(\beta)}(x)
 =-\sum_{i=1}^n
       \left[
        G_\beta\left(\frac in\right)
        -G_\beta\left(\frac{i-1}{n}\right)
       \right]x_{i:n},
\end{equation}
where
\begin{equation}\label{eq:beta-c-profile}
 G_\beta(t)=
 \begin{cases}
  \displaystyle\frac{t^\beta-\beta t}{1-\beta},&\beta\ne1,\\[5pt]
  t(1-\log t),&\beta=1,
 \end{cases}
 \qquad G_\beta(0)=0.
\end{equation}
To obtain this formula, split
$\int_0^1\min\{t/\alpha,1\}\beta\alpha^{\beta-1}d\alpha$
at $\alpha=t$.  Thus $c(t)=G_\beta(t)$.
For $\beta>1$ the spectrum is bounded by $\beta/(\beta-1)$;
for $0<\beta\leqslant1$ it is unbounded.  Nevertheless all these
estimators are uniformly estimable on a fixed bounded interval.
Their continuity moduli distinguish the power, logarithmic and
Lipschitz regimes directly.

For a different effect, let $h_H=\sum_{k=1}^H1/k$ and take
$\mu_k=1/(kh_H)$.  Since $h_H$ is unlimited, for every standard
$\eta>0$,
\[
 \sum_{k>\eta H}\mu_k
  =\frac1{h_H}\sum_{k>\eta H}\frac1k\simeq0.
\]
The numerator is limited: it is bounded, for instance, by
$1/\eta$.  Hence the Loeb pushforward is $\delta_0$ and
\begin{equation}\label{eq:harmonic-worst-case-cre}
 \st\left(
      \frac1{h_H}\sum_{k=1}^H
       \frac{{}^*\![\alpha\mapsto
                     \ES_\alpha(P_x)](k/H)}{k}
                       \right)
     =-\min_i x_i.
\end{equation}
Thus an atom at zero may arise even though every internal mixing
level is strictly positive.

Combining the uniform and harmonic vectors, with a standard
$\zeta\in[0,1]$, gives
\begin{equation}\label{eq:log-worst-mixture-cre}
 R_n^{(\zeta)}(x)
   =-\zeta\min_i x_i+(1-\zeta)R_n^{\log}(x),
 \qquad
 c_\zeta(t)=\zeta+(1-\zeta)t(1-\log t)\quad(t>0).
\end{equation}
Its lack of uniform estimability when $\zeta>0$ is exactly the
standard shadow of the infinitesimal-level mass.  This does not
preclude pointwise consistency at each fixed bounded distribution.

Finally, a hyperfinite envelope can yield an estimator which is not
comonotone.  Let $\mathscr M$ consist of the uniform vector and
the vector placing mass $1/2$ at $\lceil H/4\rceil$ and mass
$1/2$ at $H$.  Their Loeb pushforwards are Lebesgue measure and
$(\delta_{1/4}+\delta_1)/2$, respectively.  Theorem~
\ref{thm:loeb-cre-construction} gives the robust CRE
\begin{equation}\label{eq:robust-log-es-cre}
 R_n^{\mathrm{rob}}(x)
  =\max\left\{R_n^{\log}(x),
       \frac12\ES_{1/4}(P_x)
            -\frac1{2n}\sum_{i=1}^n x_i\right\}.
\end{equation}
Its rare-loss profile is
$c(t)=\max\{t(1-\log t),\frac12\min\{4t,1\}+t/2\}$,
which tends to zero.  To see that comonotonicity can fail, take
$n=4$, $x=(-1,0,0,0)$ and $y=(-1,-1,0,0)$.
These vectors are comonotone.  At $x$ the second member of the
maximum in \eqref{eq:robust-log-es-cre} is strictly larger,
whereas at $y$ the first is strictly larger.  Both members are
comonotone additive, so
$R_4^{\mathrm{rob}}(x+y)<R_4^{\mathrm{rob}}(x)+R_4^{\mathrm{rob}}(y)$.
This example illustrates the distinction in \cite{ACJP} between
law-invariant CREs represented by envelopes and comonotone CREs
represented by a single $L$-estimator.

In exact values, $R_4^{\mathrm{rob}}(x)=5/8$,
$R_4^{\mathrm{rob}}(y)=(1+\log2)/2$, and
$R_4^{\mathrm{rob}}(x+y)=3/4+\log2$.  Thus the strict gap is
$3/8-(\log2)/2>0$.

\section{Tail regularisation}\label{sec:regularisation}

Chen and Dentcheva \cite{ChenDentcheva} study smoothing and transport
stability for risk forms, regularising the reference law by convolution.
Here we modify the tail-level weights
intrinsically, keeping the underlying observations and their support fixed.

Let $\rho$ be a law-invariant coherent risk measure on $[0,1]$.
Its Fatou property is automatic. Throughout this section, we use a Kusuoka
representation
\begin{equation}\label{eq:reg-kusuoka}
 \rho(X)=\sup_{\nu\in\Gamma}\int_{[0,1]}\ES_\alpha(X)\,\nu(d\alpha),
 \qquad \ES_0(X)=-\mathop{\rm ess\,inf}X.
\end{equation}
This representation exists for every such $\rho$. Indeed, its dyadic
restrictions are replication-consistent CREs, so uniqueness in
Theorem~\ref{thm:reconstruction} identifies $\rho$ with the spectral
envelope \eqref{eq:spectral-envelope}; each step spectrum is a finite ES
mixture. This also recovers the classical Kusuoka representation
\cite{JST,Kusuoka} and includes \eqref{eq:loeb-population}.
We may take $\Gamma$ to be weakly compact: for bounded $X$, the function
$\alpha\mapsto\ES_\alpha(X)$ is continuous on $[0,1]$, so passage to the weak
closure does not change the supremum. Write
\[
 c(t)=\rho(-\mathbf1_{(0,t)}),\qquad c(0)=0,\qquad
 \kappa=\lim_{t\downarrow0}c(t).
\]
The possible jump $\kappa$ is essential in the following construction.

For $0<\varepsilon\leqslant1$, define
\begin{equation}\label{eq:reg-def}
 \rho^\varepsilon(X)
 =\sup_{\nu\in\Gamma}
   \int_{[0,1]}\ES_{\max\{\alpha,\varepsilon\}}(X)\,\nu(d\alpha).
\end{equation}
Although the display uses a Kusuoka representation, the construction is
intrinsic. Indeed, let $q_X$ be the lower quantile and put
\begin{equation}\label{eq:reg-quantile}
 q_X^\varepsilon(u)=
 \begin{cases}
  \varepsilon^{-1}\displaystyle\int_0^\varepsilon q_X(v)\,dv,
      &0<u\leqslant\varepsilon,\\
  q_X(u),&\varepsilon<u<1.
 \end{cases}
\end{equation}
This is a quantile function, and direct integration gives
$\ES_\alpha(q_X^\varepsilon)=\ES_{\max\{\alpha,\varepsilon\}}(X)$, including
$\alpha=0$. Consequently, $\rho^\varepsilon(X)=\rho(q_X^\varepsilon)$ and
\eqref{eq:reg-def} is independent of the choice of $\Gamma$.

\begin{proposition}\label{prop:tail-regularisation}
For each $\varepsilon\in(0,1]$, $\rho^\varepsilon$ is a law-invariant coherent
risk measure with the Lebesgue property. Its Bernoulli profile is
\begin{equation}\label{eq:reg-profile}
 c_\varepsilon(t)=
 \begin{cases}
  t c(\varepsilon)/\varepsilon,&0\leqslant t\leqslant\varepsilon,\\
  c(t),&\varepsilon\leqslant t\leqslant1.
 \end{cases}
\end{equation}
In particular, $\rho^\varepsilon$ is Lipschitz for $W_1$, with optimal
constant $c(\varepsilon)/\varepsilon$ on bounded laws. For every bounded $X$,
\begin{equation}\label{eq:reg-bias-pointwise}
 0\leqslant\rho(X)-\rho^\varepsilon(X)
 \leqslant\ES_0(X)-\ES_\varepsilon(X),
 \qquad \rho^\varepsilon(X)\uparrow\rho(X)
       \quad(\varepsilon\downarrow0).
\end{equation}
If $\Pa$ denotes the probability laws on an interval of length
$B=\overline x-\underline x>0$, then
\begin{equation}\label{eq:reg-bias-uniform}
 B\kappa\leqslant
 \sup_{P\in\Pa}\{\rho(P)-\rho^\varepsilon(P)\}
 \leqslant Bc(\varepsilon).
\end{equation}
Thus the left-hand supremum tends to $B\kappa$ as $\varepsilon\downarrow0$.
\end{proposition}

\begin{proof}
Each functional inside the supremum in \eqref{eq:reg-def} is a mixture of
expected shortfalls, and is therefore coherent and law invariant. The
supremum retains these properties and is finite on $L^\infty$.
The mixture indexed by $\nu$ has spectral density
\[
 \phi_{\nu,\varepsilon}(u)
  =\int_{[0,1]}
     \frac{\mathbf1_{\{u\leqslant\max\{\alpha,\varepsilon\}\}}}
          {\max\{\alpha,\varepsilon\}}\,\nu(d\alpha).
\]
These densities are non-negative and non-increasing, integrate to one, and
are bounded by
\[
 \int_{[0,1]}\frac{1}{\max\{\alpha,\varepsilon\}}\,\nu(d\alpha)
 =\frac{1}{\varepsilon}
   \int_{[0,1]}\min\{1,\varepsilon/\alpha\}\,\nu(d\alpha)
 \leqslant\frac{c(\varepsilon)}{\varepsilon},
\]
where the integrand at $\alpha=0$ is defined to be one. The quantile formula
for $W_1$ now proves the asserted Lipschitz bound. Bounded almost-sure
convergence implies convergence in $W_1$, so it also proves the Lebesgue
property.

For $t\leqslant\varepsilon$,
\[
 \min\{1,t/\max\{\alpha,\varepsilon\}\}
 =\frac{t}{\varepsilon}\min\{1,\varepsilon/\alpha\};
\]
for $t\geqslant\varepsilon$ the left-hand side equals
$\min\{1,t/\alpha\}$. Taking the supremum over $\nu$ proves
\eqref{eq:reg-profile}. Comparing $-\mathbf1_{(0,t)}$ with zero for
$0<t\leqslant\varepsilon$ proves optimality of the Lipschitz constant.

Expected shortfall decreases with its level. Thus each mixture in
\eqref{eq:reg-def} increases as $\varepsilon$ decreases and is bounded above
by the corresponding mixture in \eqref{eq:reg-kusuoka}. The discrepancy is
at most $\ES_0(X)-\ES_\varepsilon(X)$, since only levels below
$\varepsilon$ change. Taking suprema proves
\eqref{eq:reg-bias-pointwise}; its right-hand side tends to zero by bounded
convergence of lower-tail quantile averages.

For a law supported on $[\underline x,\overline x]$, each changed expected shortfall differs by
at most $B$. Hence
\[
 \rho(P)-\rho^\varepsilon(P)
 \leqslant B\sup_{\nu\in\Gamma}\nu([0,\varepsilon))
 \leqslant Bc(\varepsilon).
\]
For the converse bound, take $P_t=t\delta_{\underline x}+(1-t)\delta_{\overline x}$ with
$0<t\leqslant\varepsilon$. Cash additivity, positive homogeneity and
\eqref{eq:reg-profile} give
\[
 \rho(P_t)-\rho^\varepsilon(P_t)
   =B\left(c(t)-\frac{t}{\varepsilon}c(\varepsilon)\right).
\]
Letting $t\downarrow0$ proves the lower bound in
\eqref{eq:reg-bias-uniform}.
\end{proof}

The construction has a particularly simple finite form. Suppose that a
coherent estimator on $\mathbb R^n$ is represented by a compact set
$A\subseteq\{a_1\geqslant\cdots\geqslant a_n\geqslant0:
\sum_j a_j=1\}$ as
\[
 R_n(x)=\max_{a\in A}\left(-\sum_{j=1}^n a_jx_{j:n}\right).
\]
For $1\leqslant k\leqslant n$, replace $a$ by
\begin{equation}\label{eq:reg-finite-weights}
 a_j^{[k]}=
 \begin{cases}
  k^{-1}\sum_{i=1}^k a_i,&1\leqslant j\leqslant k,\\
  a_j,&k<j\leqslant n.
 \end{cases}
\end{equation}
The resulting weights are still non-negative, non-increasing and normalised.
They therefore define a coherent estimator
\[
 R_n^{[k]}(x)=\max_{a\in A}\left(-\sum_{j=1}^n a_j^{[k]}x_{j:n}\right).
\]
Writing the discrete Kusuoka weights as
$b_j=j(a_j-a_{j+1})$, where $a_{n+1}=0$, shows that this operation replaces
each level $j/n$ by $\max\{j/n,k/n\}$. Consequently, if $R_n$ is the
restriction of $\rho$ to equiprobable laws, then
\begin{equation}\label{eq:reg-finite-identification}
 R_n^{[k]}(x)=\rho^{k/n}\left(\frac1n\sum_{i=1}^n\delta_{x_i}\right).
\end{equation}
For the worst-case estimator $R_n(x)=-x_{1:n}$, the construction gives
$R_n^{[k]}(x)=-k^{-1}\sum_{j=1}^k x_{j:n}$. For the mean estimator it makes
no change. Equivalently, evaluate the original CRE on the ordered sample after
replacing its first $k$ observations by their arithmetic mean. This
computes the regularisation even when the dual weights are not explicitly
available. The construction lowers capital: $R_n^{[k]}(x)\leqslant R_n(x)$.
It is therefore an approximation with a quantified bias, not an automatic
conservative correction.

For the logarithmic estimator in \eqref{eq:logarithmic-cre}, the
regularised rule is the $L$-estimator with increments of
\begin{equation}\label{eq:log-floor-profile}
 G_\varepsilon(t)=
 \begin{cases}
 t(1-\log\varepsilon),&0\leqslant t\leqslant\varepsilon,\\
 t(1-\log t),&\varepsilon\leqslant t\leqslant1.
 \end{cases}
\end{equation}
Its spectrum equals $1-\log\varepsilon$ below $\varepsilon$ and
$-\log u$ above $\varepsilon$. Thus the logarithmic construction yields
an explicit family of bounded-spectrum CREs with optimal Lipschitz
constant $1-\log\varepsilon$. Section~\ref{sec:catastrophe} compares
the original and regularised estimators by exact finite-sample calculations.

Pointwise strong consistency of coherent plug-in estimators is classical;
see Wozabal and Wozabal \cite[Lemma~3.1 and Theorem~3.2]{WozabalWozabal}.
The following bounded-support version of their uniform expected-shortfall
argument is recalled to handle a regularisation level that varies with the
sample size. Related consistency results on larger spaces and for
$L$-estimators appear in \cite{GaoXanthos} and
\cite[Section~4]{ACJP}, respectively.

\begin{lemma}\label{lem:uniform-es-consistency}
Let $X_1,X_2,\ldots$ be independent with common bounded law $P$, and let
$\widehat P_n=n^{-1}\sum_{i=1}^n\delta_{X_i}$. Then
\begin{equation}\label{eq:uniform-es-consistency}
 \mathcal E_n:=\sup_{0\leqslant\alpha\leqslant1}
       |\ES_\alpha(\widehat P_n)-\ES_\alpha(P)|
       \longrightarrow0\qquad\text{almost surely}.
\end{equation}
\end{lemma}

\begin{proof}
On an event of probability one, $W_1(\widehat P_n,P)\to0$ and
$\min_{1\leqslant i\leqslant n}X_i\to\operatorname{ess\,inf}X_1$.
The first assertion follows from the Glivenko--Cantelli theorem and the
common support bound. For the second, each interval immediately above the
essential infimum has positive $P$-probability and is eventually visited;
apply this observation to a countable sequence of such intervals.
Consequently, on this event,
$f_n(\alpha):=\ES_\alpha(\widehat P_n)$ converges to
$f(\alpha):=\ES_\alpha(P)$ at $\alpha=0$ and at every
$\alpha>0$, the latter by
\[
 |f_n(\alpha)-f(\alpha)|\leqslant
            \frac{W_1(\widehat P_n,P)}{\alpha}.
\]
The functions $f_n$ are non-increasing, and $f$ is continuous on $[0,1]$.
Given $\eta>0$, choose a finite partition of $[0,1]$ on each of whose
intervals the oscillation of $f$ is at most $\eta$.
Monotonicity bounds the error at every point by the largest error at a
partition endpoint plus $\eta$. Taking the limit superior in $n$ and then
letting $\eta\downarrow0$ proves \eqref{eq:uniform-es-consistency}.
\end{proof}

For the unregularised risk, the same Kusuoka representation yields
$|\rho(\widehat P_n)-\rho(P)|\leqslant\mathcal E_n$.
Thus ordinary coherent plug-in estimators are already strongly consistent
at every fixed bounded law, including when $\kappa>0$.

\begin{corollary}\label{cor:reg-estimation}
Let $X_1,X_2,\ldots$ be independent with common law $P\in\Pa$,
let $\widehat P_n=n^{-1}\sum_{i=1}^n\delta_{X_i}$, and write $B=\overline x-\underline x$.
For every $\delta\in(0,1)$, with probability at least $1-\delta$,
\begin{equation}\label{eq:reg-confidence}
 |\rho^\varepsilon(\widehat P_n)-\rho^\varepsilon(P)|
 \leqslant Bc_\varepsilon\left(
   \min\left\{1,\sqrt{\frac{\log(2/\delta)}{2n}}\right\}\right).
\end{equation}
On the same event, the error relative to $\rho(P)$ is at most
$Bc(\varepsilon)+Bc_\varepsilon(u_{n,\delta})$.

For every deterministic sequence $\varepsilon_n\in(0,1]$ tending to zero,
\begin{equation}\label{eq:reg-strong-consistency}
 \rho^{\varepsilon_n}(\widehat P_n)\longrightarrow\rho(P)
       \qquad\text{almost surely}.
\end{equation}
In particular, the coherent estimators in
\eqref{eq:reg-finite-weights} are strongly consistent whenever
$1\leqslant k_n\leqslant n$ and $k_n/n\to0$.
This conclusion allows $\kappa>0$ and imposes no lower growth condition
on $k_n$.

The estimators in \eqref{eq:reg-strong-consistency} are uniformly
consistent over $\Pa$ if and only if $\kappa=0$.
\end{corollary}

\begin{proof}
The Dvoretzky--Kiefer--Wolfowitz inequality gives
$\|F_{\widehat P_n}-F_P\|_\infty\leqslant u_{n,\delta}$
with probability at least $1-\delta$.
On the common support interval,
$W_1(\widehat P_n,P)\leqslant Bu_{n,\delta}$.
Apply the exact bounded-support modulus theorem to $\rho^\varepsilon$
and use \eqref{eq:reg-profile} to obtain \eqref{eq:reg-confidence}.
The bound relative to $\rho(P)$ then follows from
\eqref{eq:reg-bias-uniform}.

The Kusuoka representation gives, simultaneously for all
$\varepsilon\in(0,1]$,
\[
 |\rho^\varepsilon(\widehat P_n)-\rho^\varepsilon(P)|\leqslant \mathcal E_n.
\]
Indeed, each integrand is evaluated at a level in $[0,1]$, and taking
probability mixtures and suprema cannot increase the uniform error.
By \eqref{eq:reg-bias-pointwise},
\[
 |\rho^{\varepsilon_n}(\widehat P_n)-\rho(P)|
 \leqslant \mathcal E_n+\ES_0(P)-\ES_{\varepsilon_n}(P).
\]
Both terms tend to zero, the first almost surely by
Lemma~\ref{lem:uniform-es-consistency}, and the second by continuity of
the population expected-shortfall curve at zero. This proves
\eqref{eq:reg-strong-consistency}; the assertion about integer sequences
follows from \eqref{eq:reg-finite-identification}.

The inequality $\rho^\varepsilon\leqslant\rho$ implies
$c_\varepsilon\leqslant c$ by evaluation on negative indicators; it does
not require concavity of $c$. If $\kappa=0$, the finite-sample
bound relative to the original target is at most
$B c(\varepsilon_n)+B c(u_{n,\delta})$, which tends to zero for every
fixed $\delta>0$, uniformly over $P\in\Pa$.
Letting $\delta\downarrow0$ proves uniform consistency.
If $\kappa>0$, Theorem~\ref{thm:statistics} rules out uniform consistency
for every sequence of estimators.
\end{proof}

For a fixed $\varepsilon$, the estimators
$x\mapsto\rho^\varepsilon(n^{-1}\sum_i\delta_{x_i})$ are replication
consistent, because replication preserves the empirical law. A changing
regularisation level need not have this property. For example, with the
worst-case risk and $k_n=\lceil n^{3/4}\rceil$, the vector $(0,1)$ has
regularised risk $-1/2$, whereas its duplicate $(0,0,1,1)$ has regularised
risk $-1/3$. The corollary concerns a sequence of individually coherent
estimators; it does not assert that this sequence satisfies the fixed-rule
replication axiom.

The regularisation can be carried out before taking standard parts in
\eqref{eq:internal-estimator-shadow}: replace $k/H$ by
$\max\{k/H,\varepsilon\}$. For standard $\varepsilon>0$, Loeb integration
then gives precisely the CRE $x\mapsto\rho^\varepsilon(P_x)$.
If the floor $\eta>0$ is infinitesimal, the shadow of every fixed standard
sample is unchanged. Indeed, whenever a level is altered, both the old
level and $\eta$ lie below $1/n$, where empirical expected shortfall is
identically $-x_{1:n}$. A positive standard floor and an infinitesimal
floor consequently have different finite-sample effects.

\section{Data-driven risk-averse optimisation}\label{sec:retention}

The stability modulus controls empirical decisions over an entire compact
set, without a union bound over the candidate decisions. We first treat a
common scalar scenario; a transport formulation below allows general
scenario spaces.

Let $\mathcal D$ be a non-empty compact metric decision set, let $Z$ take values in a
compact interval $I$, and let $\ell,g:\mathcal D\times I\to\R$ be continuous. Write
$P$ for the law of $Z$ and set
\[
 J_P(d)=\rho(-\ell(d,Z))+\E_P g(d,Z),\qquad
 J_P^\varepsilon(d)=\rho^\varepsilon(-\ell(d,Z))
                         +\E_P g(d,Z).
\]
Here $\ell$ is a retained loss and $g$ an additional cost. Put
\[
 B(d)=\max_{z\in I}\ell(d,z)-\min_{z\in I}\ell(d,z),\qquad
 C(d)=\max_{z\in I}g(d,z)-\min_{z\in I}g(d,z).
\]
These functions are continuous on $\mathcal D$. For notational convenience set
$\rho^0=\rho$, $c_0=c$ and $b_0=0$. For $\varepsilon>0$, let
$b_\varepsilon$ be a deterministic bound such that
\begin{equation}\label{eq:decision-bias-assumption}
 0\leqslant\rho(X)-\rho^\varepsilon(X)
       \leqslant \operatorname{osc}(X)b_\varepsilon
 \quad\hbox{for every bounded }X.
\end{equation}
Here $\operatorname{osc}(X)=\operatorname{ess\,sup}X-\operatorname{ess\,inf}X$.
Proposition~\ref{prop:tail-regularisation} permits
$b_\varepsilon=c(\varepsilon)$. Let $\widehat J_n^\varepsilon$ denote
$J_{\widehat P_n}^\varepsilon$, based on the same independent observations
for every decision.

\begin{theorem}\label{thm:decision-stability}
Suppose that, for each $d\in \mathcal D$, the functions $\ell(d,\cdot)$ and
$g(d,\cdot)$ are monotone, in either direction. Let
$u=u_{n,\delta}$ be as in Theorem~\ref{thm:statistics}, and define
\[
 s_\varepsilon(d)=B(d)c_\varepsilon(u)+C(d)u,
 \qquad 0\leqslant\varepsilon\leqslant1.
\]
With probability at least $1-\delta$, simultaneously for every
$d\in \mathcal D$ and $\varepsilon\in[0,1]$,
\begin{equation}\label{eq:decision-uniform}
 |\widehat J_n^\varepsilon(d)-J_P^\varepsilon(d)|
       \leqslant s_\varepsilon(d).
\end{equation}
Let $d_*$ minimise $J_P$. On this event, every empirical
$\eta$-minimiser $\widehat d_n^\varepsilon$, where $\eta\geqslant0$,
satisfies the original-target regret bound
\begin{equation}\label{eq:decision-regret}
 J_P(\widehat d_n^\varepsilon)-J_P(d_*)
 \leqslant B(\widehat d_n^\varepsilon)b_\varepsilon
       +s_\varepsilon(\widehat d_n^\varepsilon)
       +s_\varepsilon(d_*)+\eta.
\end{equation}
More generally, every candidate $d\in \mathcal D$ has the observable certificate
\begin{equation}\label{eq:decision-certificate}
 J_P(d)-\min_{e\in \mathcal D}J_P(e)
 \leqslant \widehat J_n^\varepsilon(d)
       +s_\varepsilon(d)+B(d)b_\varepsilon
       -\inf_{e\in \mathcal D}
          \{\widehat J_n^\varepsilon(e)-s_\varepsilon(e)\}.
\end{equation}
Thus the same event remains valid when the regularisation level and the
decision are selected from the data.
\end{theorem}

\begin{proof}
On the Dvoretzky--Kiefer--Wolfowitz event
$\|F_{\widehat P_n}-F_P\|_\infty\leqslant u$, every monotone transform of
$Z$ has empirical and population distribution functions differing by at
most $u$. Indeed, the inverse image of a half-line is an interval having
an endpoint at an end of $I$; open endpoints are handled by left limits
of the distribution functions. Consequently,
\[
 W_1\bigl(\mathcal L_{\widehat P_n}(-\ell(d,Z)),
          \mathcal L_P(-\ell(d,Z))\bigr)\leqslant B(d)u.
\]
The exact modulus gives a risk error at most $B(d)c_\varepsilon(u)$.
Layer-cake integration gives an expectation error at most $C(d)u$.
If either oscillation is zero, the corresponding error is zero. These
arguments are deterministic on the single distributional event and
apply to every $d$ and every $\varepsilon$, proving
\eqref{eq:decision-uniform}.

Compactness and joint continuity imply continuity of $J_P$ and of the
empirical criteria in $d$: coherent risk measures are $1$-Lipschitz in
the supremum norm. Hence their minima exist. Insert empirical
$\eta$-optimality between \eqref{eq:decision-uniform} and
\eqref{eq:decision-bias-assumption}. Only the bias at the selected decision
is needed, because $J_P^\varepsilon(d_*)\leqslant J_P(d_*)$; this proves
\eqref{eq:decision-regret}. Finally,
\[
 J_P(d)\leqslant\widehat J_n^\varepsilon(d)
                       +s_\varepsilon(d)+B(d)b_\varepsilon,
 \qquad
 J_P(e)\geqslant\widehat J_n^\varepsilon(e)-s_\varepsilon(e).
\]
Taking the infimum in the second inequality proves
\eqref{eq:decision-certificate}.
\end{proof}

A Borel approximate optimiser can be chosen explicitly. For fixed
$\varepsilon$ and every $\eta>0$, choose a countable
dense sequence $(d_j)$ in $\mathcal D$ and take its first member satisfying
\[
 \widehat J_n^\varepsilon(d_j)
     <\inf_k\widehat J_n^\varepsilon(d_k)+\eta.
\]
This gives a Borel empirical $\eta$-minimiser, since each sample criterion
is continuous and the infimum is countable. One may similarly compare
any finite collection of floors using the certificates
\eqref{eq:decision-certificate}, choosing the first at a tie. This gives
a measurable data-dependent choice with the same confidence level.
Writing $\overline B=\max_{\mathcal D} B$ and
$\overline C=\max_{\mathcal D} C$, the simpler
bound
\[
 J_P(\widehat d_n^\varepsilon)-\min_{d\in\mathcal D}J_P(d)
 \leqslant \overline B\{b_\varepsilon+2c_\varepsilon(u)\}
             +2\overline C u+\eta
\]
exhibits the bias--stability trade-off. If $\kappa=0$, any deterministic
$\varepsilon_n\downarrow0$, $u_n\downarrow0$ and $\eta_n\downarrow0$
make its right-hand side tend to zero when $b_\varepsilon=c(\varepsilon)$, using
$c_{\varepsilon_n}(u_n)\leqslant c(u_n)$.

For the logarithmic spectrum, the sharp bias constant
$b_\varepsilon=\varepsilon/e$, proved below, makes the trade-off explicit.
If $0<u\leqslant1/(2e)$, taking $\varepsilon=2eu$ gives
\[
 \overline B\{b_\varepsilon+2c_\varepsilon(u)\}
       =2\overline B u\{1-\log(2u)\},
\]
which is smaller than the unregularised risk term
$2\overline B u(1-\log u)$ by $2\overline B u\log2$.
The expected-cost term is unchanged. This comparison concerns the proved
regret bounds; it does not assert that regularisation improves the actual
regret for every underlying law.

There is also a deterministic transport version. Let the scenario space
be a compact metric space, let $\ell$ and $g$ be jointly continuous, and
suppose their scenario Lipschitz constants at $d$ admit continuous bounds
$L_\ell(d)$ and $L_g(d)$. On any event
$W_1(\widehat P_n,P)\leqslant r$, the same deterministic inequalities hold with
\begin{equation}\label{eq:decision-transport-error}
 s_\varepsilon(d;r)
 =B(d)c_\varepsilon\!\left(
       \min\left\{1,\frac{L_\ell(d)r}{B(d)}\right\}\right)
       +\min\{C(d),L_g(d)r\},
\end{equation}
where the first term is zero when $B(d)=0$. This follows by pushing
forward a coupling for $W_1$ and applying the exact modulus; the mean
term follows from the same coupling. In particular the risk term is at
most $L_\ell(d)r\,c(\varepsilon)/\varepsilon$ when
$\varepsilon>0$. Thus the regularisation replaces a possibly discontinuous
transport modulus by its optimal linear bound, while the regret still
refers to the original objective.

\subsection{Stop-loss retention}

The same distributional event that controls capital estimation also controls
the choice of an insurance contract. Let a loss $L$ take values in $[0,M]$.
At retention $d\in[0,M]$, the insurer retains $L\wedge d$ and cedes
$(L-d)_+$. For a fixed loading $\theta\geqslant0$, consider the
risk-adjusted cost
\[
 J_P(d)=\rho(-(L\wedge d))+(1+\theta)\E_P(L-d)_+.
\]
Thus the ceded loss is priced by its loaded expectation and the retained
loss by $\rho$. Both terms are estimated from the same independent sample.
Write $\widehat J_n$ for the resulting empirical criterion, and let
$\widehat d_n$ be a measurable minimiser on $[0,M]$. Minimisers exist because
the criterion is continuous in $d$. Its joint continuity in the sample
vector and $d$ makes the smallest minimiser Borel measurable: for each
$r\in[0,M]$, the event $\{\widehat d_n\leqslant r\}$ is the equality
set of the continuous minima over $[0,r]$ and $[0,M]$.

\begin{proposition}\label{prop:retention}
Let $u=u_{n,\delta}$ be as in Theorem~\ref{thm:statistics}, and put
\[
 s(d)=d c(u)+(1+\theta)(M-d)u.
\]
With probability at least $1-\delta$, simultaneously for every
$d\in[0,M]$,
\begin{equation}\label{eq:retention-uniform}
 |\widehat J_n(d)-J_P(d)|\leqslant s(d).
\end{equation}
On this event, for every population minimiser $d_*$,
\begin{equation}\label{eq:retention-regret}
 J_P(\widehat d_n)-J_P(d_*)
 \leqslant s(\widehat d_n)+s(d_*)
 \leqslant2M\max\{c(u),(1+\theta)u\}.
\end{equation}
Suppose instead that the empirical criterion uses $\rho^\varepsilon$, and
let $\widehat d_n^\varepsilon$ minimise this criterion. If
\[
 0\leqslant\rho(-Y)-\rho^\varepsilon(-Y)\leqslant d b_\varepsilon
 \quad\hbox{whenever }0\leqslant Y\leqslant d,
\]
then, on the same event, the regret relative to the original criterion is
at most
\begin{equation}\label{eq:retention-regret-reg}
 \widehat d_n^\varepsilon b_\varepsilon
 +s_\varepsilon(\widehat d_n^\varepsilon)+s_\varepsilon(d_*),
 \qquad
 s_\varepsilon(d)=d c_\varepsilon(u)+(1+\theta)(M-d)u.
\end{equation}
One may always take $b_\varepsilon=c(\varepsilon)$.
For the logarithmic spectrum, the optimal value is
$b_\varepsilon=\varepsilon/e$.
\end{proposition}

\begin{proof}
Apply Theorem~\ref{thm:decision-stability} with
$\mathcal D=I=[0,M]$, $\ell(d,z)=z\wedge d$ and
$g(d,z)=(1+\theta)(z-d)_+$. Their oscillations are
$B(d)=d$ and $C(d)=(1+\theta)(M-d)$, so
\eqref{eq:decision-uniform}--\eqref{eq:decision-regret} give all the stated
sampling and regret bounds, with $\eta=0$.
Proposition~\ref{prop:tail-regularisation} permits
$b_\varepsilon=c(\varepsilon)$.
For the logarithmic spectrum, Fubini's theorem gives
\[
 \rho(-Y)-\rho^\varepsilon(-Y)
 =\int_0^d\{G(P(Y>t))-G_\varepsilon(P(Y>t))\}\,dt.
\]
For $0<v<\varepsilon$, the integrand's profile difference is
$v\log(\varepsilon/v)$; it vanishes for $v\geqslant\varepsilon$ and
has maximum $\varepsilon/e$. Equality holds for
$Y=d\ind_A$ and $\lambda(A)=\varepsilon/e$, which proves optimality.
\end{proof}

The observable certificate \eqref{eq:decision-certificate} applies with
the same $s_\varepsilon(d)$ and bias $d b_\varepsilon$. A numerical
optimiser with objective error at most $\eta$ adds at most $\eta$ to the
regret bounds.

\section{A bounded catastrophe loss}
\label{sec:catastrophe}

Let $\xi_p$ have the Bernoulli distribution with parameter $p\in(0,1)$, and
let $X=-M\xi_p$, where $M>0$ is the loss severity.  This two-point model allows
the finite-sample errors to be computed directly.  We consider expectation,
expected shortfall at tail probability $\alpha$, the spectrum
$\phi(u)=1/(2\sqrt{u})$, and worst-case risk.  Their values at $X$ are
$M h(p)$, with, respectively,
\[
 h(t)=t,\qquad h(t)=\min\{t/\alpha,1\},\qquad
 h(t)=\sqrt{t},\qquad h(t)=\mathbf1_{(0,1]}(t).
\]
For the first three functionals, $h(t)$ is the integral of the spectrum
over $(0,t)$; worst-case risk is $-\operatorname{ess\,inf}X$.
We also include the logarithmic estimator from
\eqref{eq:logarithmic-cre}, with $h(t)=G(t)=t(1-\log t)$ and $G(0)=0$,
and its regularised version with $h=G_\varepsilon$ from
\eqref{eq:log-floor-profile}. These are the additional estimators obtained
from the hyperfinite uniform mixture and its tail-level floor.
The corresponding finite estimators are
\begin{equation}\label{eq:catastrophe-finite-risk}
 R_n(x)=-\sum_{i=1}^n
 \left(h\left(\frac{i}{n}\right)-h\left(\frac{i-1}{n}\right)\right)x_{i:n}.
\end{equation}
The increments are non-negative, non-increasing, and sum to one, so these
are coherent, law-invariant, comonotonically additive estimators.
If $\mathsf D_kx$ repeats each coordinate of $x$ exactly $k$ times, telescoping gives
$R_{kn}(\mathsf D_kx)=R_n(x)$.  They are therefore compatible under replication.

For an independent sample of size $n$, let $K$ be the number of
catastrophes.  Then $K\sim\operatorname{Bin}(n,p)$, and
\eqref{eq:catastrophe-finite-risk} gives the exact identity
$\widehat\rho_{n,h}=Mh(K/n)$.
Writing $b_{n,k}(p)=\binom nkp^k(1-p)^{n-k}$, we obtain
\begin{align}
 \frac{\operatorname{Bias}(\widehat\rho_{n,h})}{M}
 &=\sum_{k=0}^n b_{n,k}(p)\bigl(h(k/n)-h(p)\bigr),
 \label{eq:catastrophe-bias}\\
 \frac{\operatorname{RMSE}(\widehat\rho_{n,h})^2}{M^2}
 &=\sum_{k=0}^n b_{n,k}(p)\bigl(h(k/n)-h(p)\bigr)^2.
 \label{eq:catastrophe-rmse}
\end{align}
Here $\operatorname{RMSE}$ denotes the square root of the mean squared
error.  Concavity of $h$ and Jensen's inequality show that each bias is
non-positive.  For expectation, the bias is zero and the normalised
RMSE is $\sqrt{p(1-p)/n}$.  For worst-case risk, the estimator equals zero
if $K=0$ and equals $M$ otherwise, whence
\begin{equation}\label{eq:catastrophe-worst-case}
 \frac{\operatorname{Bias}(\widehat\rho_{n,h})}{M}=-(1-p)^n,
 \qquad
 \frac{\operatorname{RMSE}(\widehat\rho_{n,h})}{M}=(1-p)^{n/2}.
\end{equation}

Table~\ref{tab:catastrophe} evaluates the finite sums for $p=0.005$ and
$\alpha=0.025$.  The population risks divided by $M$ are $0.005$, $0.2$,
$\sqrt{0.005}\approx0.070711$, and $1$, respectively.
For the logarithmic estimator and its floor $\varepsilon=0.025$, the
normalised targets are approximately $0.031492$ and $0.023444$.
The entries are numerical evaluations of the binomial formulae, with no
simulation or fitted parameters.

\begin{table}[htbp]
\caption{Bias and RMSE as fractions of the catastrophe severity $M$.
The notation $|\cdot|<10^{-20}$ bounds the absolute bias; all biases are
non-positive.}
\label{tab:catastrophe}
\centering
\begin{tabular}{lrrr}
\toprule
Risk & $n$ & Bias$/M$ & RMSE$/M$\\
\midrule
Expectation & 32 & $0$ & $0.012469$\\
 & 128 & $0$ & $0.006234$\\
 & 512 & $0$ & $0.003117$\\
 & 2048 & $0$ & $0.001559$\\[2pt]
Expected shortfall & 32 & $-0.051802$ & $0.359053$\\
 & 128 & $-0.001192$ & $0.244691$\\
 & 512 & $-8.929\cdot10^{-8}$ & $0.124687$\\
 & 2048 & $|\cdot|<10^{-20}$ & $0.062344$\\[2pt]
Power spectrum ($\gamma=1/2$) & 32 & $-0.043659$ & $0.078576$\\
 & 128 & $-0.023051$ & $0.057095$\\
 & 512 & $-0.005147$ & $0.026980$\\
 & 2048 & $-0.000902$ & $0.011296$\\[2pt]
Worst-case risk & 32 & $-0.851802$ & $0.922931$\\
 & 128 & $-0.526447$ & $0.725566$\\
 & 512 & $-0.076810$ & $0.277146$\\
 & 2048 & $-0.000035$ & $0.005900$\\[2pt]
Logarithmic spectrum & 32 & $-0.009684$ & $0.054179$\\
 & 128 & $-0.004176$ & $0.032374$\\
 & 512 & $-0.001082$ & $0.016692$\\
 & 2048 & $-0.000247$ & $0.008271$\\[2pt]
Regularised logarithmic spectrum & 32 & $-0.001637$ & $0.053332$\\
 & 128 & $-0.000034$ & $0.029080$\\
 & 512 & $-2.317\cdot10^{-9}$ & $0.014616$\\
 & 2048 & $|\cdot|<10^{-20}$ & $0.007308$\\

\bottomrule
\end{tabular}
\end{table}

For $n=32$, $128$, $512$, and $2048$, the probabilities of observing no
catastrophe are $0.851802$, $0.526447$, $0.076810$, and $0.000035$, respectively.
Expected shortfall is nearly unbiased by $n=512$, although its RMSE is
still about $0.125M$ compared with a population risk of $0.2M$.
Thus small bias alone gives little information about the accuracy of the
capital estimate.  The risks in the table have different population
targets, so their absolute errors do not rank the risks against one another.

The fixed floor changes the population target: here
$G(p)-G_{0.025}(p)\approx0.008047$. To assess the resulting trade-off,
Table~\ref{tab:common-target} compares both estimators with the same
unregularised target $MG(p)$. Its entries use
\eqref{eq:catastrophe-bias}--\eqref{eq:catastrophe-rmse}, with $h(k/n)$
replaced by $G_\varepsilon(k/n)$ but $h(p)$ retained as $G(p)$ in the
regularised columns.

\begin{table}[htbp]
\caption{Errors relative to the same logarithmic target, divided by $M$;
$p=0.005$ and $\varepsilon=0.025$.}
\label{tab:common-target}
\centering
\begin{tabular}{rrrrr}
\toprule
&\multicolumn{2}{c}{Original estimator}
&\multicolumn{2}{c}{Regularised estimator}\\
$n$ & Bias & RMSE & Bias & RMSE\\
\midrule
32 & $-0.009684$ & $0.054179$ & $-0.009684$ & $0.054179$\\
128 & $-0.004176$ & $0.032374$ & $-0.008081$ & $0.030182$\\
512 & $-0.001082$ & $0.016692$ & $-0.008047$ & $0.016685$\\
2048 & $-0.000247$ & $0.008271$ & $-0.008047$ & $0.010870$\\
\bottomrule
\end{tabular}
\end{table}

At $n=32$, both sample maps agree at every possible catastrophe count,
since every positive empirical frequency exceeds $0.025$. At $n=128$,
the fixed floor reduces RMSE against the original target; by $n=2048$,
its persistent bias outweighs the reduction in sampling error.
The vanishing-floor schedule in Corollary~\ref{cor:reg-estimation}
recovers the original target, whereas a fixed floor need not do so.

The response to one observed catastrophe has a particularly simple form:
\[
 R_n(-Me_1)=Mh(1/n).
\]
For $n=2^m$, one has $r_m=R_{2^m}$, so the normalised response
$r_m(-Me_1)/M$ is exactly $h(2^{-m})$.
It tends to zero for all the displayed risks except worst-case risk,
for which it remains one. For each fixed $p>0$, all six plug-in
estimators are strongly consistent: the strong law gives $K/n\to p$,
and each $h$ is continuous at $p$.  The distinction concerns uniform
control over rare events.  Indeed, if $p_n=1/n$, worst-case risk remains
equal to $M$ while the probability of a zero estimate is
$(1-1/n)^n\to e^{-1}$.  Its normalised RMSE consequently tends to
$e^{-1/2}$.  This calculation makes explicit how replication compatibility
can coexist with a failure of uniform stability as catastrophe
probabilities approach zero.

For the retention problem of Section~\ref{sec:retention}, the same
catastrophe model gives
\[
 J_P(d)=d c(p)+(1+\theta)(M-d)p.
\]
The optimum is full reinsurance ($d=0$) when $c(p)>(1+\theta)p$,
and full retention ($d=M$) when the inequality is reversed. For the
logarithmic spectrum, the threshold is $\tau=e^{-\theta}$.
Selecting full reinsurance at a tie, the empirical rule retains the
whole loss exactly when $K/n>\tau$. Its expected regret, divided by $M$,
is
\[
 |G(p)-(1+\theta)p|
 \begin{cases}
  P(K/n>\tau),&p<\tau,\\
  P(K/n\leqslant\tau),&p>\tau,
 \end{cases}
\]
and is zero at $p=\tau$. These are finite binomial sums. A floor
$\varepsilon\leqslant\tau$ leaves every empirical decision unchanged.
If instead $\varepsilon>\tau$, the regularised rule selects full
retention whenever $K>0$. For $p<\tau$, its expected regret therefore
tends to the positive value $M\{G(p)-(1+\theta)p\}$. Regularisation
can consequently alter the preferred contract even when its sampling
error tends to zero; the original-target bias in
\eqref{eq:retention-regret-reg} accounts for this possibility.

\appendix
\section{A hyperfinite representation}\label{sec:hyperfinite}

We record the nonstandard form of Theorem~\ref{thm:reconstruction}.
For nonstandard extensions, transfer and standard parts, see
\cite{AFHL,Robinson}.

In a nonstandard extension with transfer and infinite hyperintegers,
fix standard $X\in\Linf$ and define the standard sequences
\[
 \mathbf x_X(m)=P_mX,\qquad \mathbf D(m)=D_m,
 \qquad \mathbf K(m)=K_m.
\]
For $\mathfrak m\in{}^*\N\setminus\N$, put $N=2^{\mathfrak m}$,
$x_{\mathfrak m}^\#=({}^*\mathbf x_X)(\mathfrak m)$,
$D_{\mathfrak m}^\#=({}^*\mathbf D)(\mathfrak m)$ and
$K_{\mathfrak m}^\#=({}^*\mathbf K)(\mathfrak m)$.
Both dual sets are internal and consist of internal probability vectors
of hyperfinite length $N$. The sets themselves need not be hyperfinite.

\begin{proposition}\label{prop:hyperfinite}
For every standard $X\in\Linf[0,1]$ and every infinite $\mathfrak m$ as above,
\begin{align}
 \rho(X)
 &=\st\left(\max_{a\in D_{\mathfrak m}^\#}
                         \sum_{j=1}^{N}a_j(-x_{\mathfrak m,j}^\#)\right)
 \label{eq:hyperfinite-dual}\\
 &=\st\left(\max_{a\in K_{\mathfrak m}^\#}
                         \sum_{j=1}^{N}a_j(-x_{\mathfrak m,j:N}^\#)\right).
 \label{eq:hyperfinite-sorted}
\end{align}
Here $(x_{\mathfrak m,j:N}^\#)$ is the internally sorted vector.
\end{proposition}

\begin{proof}
The finite support identities in Lemma~\ref{lem:finite} transfer to
index $\mathfrak m$, including existence of the maxima and internal sorting.
The scalar sequence $v_m=r_m(P_mX)$ converges to $\rho(X)$ by
Theorem~\ref{thm:reconstruction} and satisfies
$|v_m|\leqslant\|X\|_\infty$. The nonstandard criterion for convergence
therefore gives $({}^*\mathbf v)(\mathfrak m)\simeq\rho(X)$. Taking its
standard part proves both formulae.
\end{proof}

Only the specified finite models are transferred; no identification of
the Lebesgue and full hyperfinite Loeb measure algebras is used
\cite{KeislerSun}.

\subsection*{Declarations}

\noindent\textit{Data and computation.}
Tables~\ref{tab:catastrophe} and~\ref{tab:common-target} evaluate
\eqref{eq:catastrophe-bias}--\eqref{eq:catastrophe-rmse} without empirical data.
The supplementary source code \cite{KaniaCode} reproduces both tables,
checks the worked examples with high-precision arithmetic, and evaluates
the exact binomial retention-regret formulae. The accompanying archive
includes the script, dependency list and instructions.

\medskip
\noindent\textit{AI assistance.}
OpenAI's ChatGPT assisted with literature searches and bibliographic checks,
restructuring the exposition, proposing and checking proof details and risk
and regret bounds, writing numerical verification code, and preparing the
LaTeX source. The author assumes responsibility for all content.

\end{document}